\documentclass{sig-alternate-05-2015}
\usepackage{cite}
\usepackage{graphicx}
\usepackage{subfigure}
\usepackage{epstopdf}
\usepackage{multirow}
\usepackage{algorithmicx}
\usepackage{algorithm}
\usepackage[noend]{algpseudocode}
\usepackage{times}

\newdef{definition}{Definition}

\newtheorem{lemma}{Lemma}
\newtheorem{skippingrule}{Skipping Rule}

\newcommand{\eucdist}[2]{\|#1 \; #2\|}
\newcommand{\eucdistmin}[2]{\|#1 \; #2\|_{\mathit{min}}}

\newcommand{\stitle}[1]{\vspace{0.8ex}\noindent{\bf #1.}$ $}

\newcounter{example}[section]
\renewcommand{\theexample}{\nthesection.\arabic{example}}
\newenvironment{example}{
	\refstepcounter{example}
	{\noindent\bf  Example  \theexample:}}{
	\eop\vspace{1ex}}

\newcommand{\eop}{\hspace*{\fill}\mbox{$\Box$}}
\newcommand{\nthesection}{\arabic{section}}

\begin{document}






%

\title{A Density-Based Approach to the Retrieval of Top-K Spatial Textual Clusters\titlenote{This research was  supported in part by project 61502310 from National Natural Science Foundation of China and a grant from the Obel Family Foundation.}}
%
%
%
%
%
\numberofauthors{2} 

\author{
	%
	%
	\alignauthor
	Dingming Wu\\
	\affaddr{College of Computer Science \& Software Engineering, Shenzhen University}\\
	\affaddr{Shenzhen, China}\\
	\email{dingming@szu.edu.cn}
	\alignauthor
	Christian S. Jensen\\
	\affaddr{Department of Computer Science}\\
	\affaddr{Aalborg University}\\
	\affaddr{Aalborg, Denmark}\\
	\email{csj@cs.aau.dk}
}

\maketitle
\begin{abstract}
Keyword-based web queries with local intent retrieve web content that
is relevant to supplied keywords and that represent points of
interest that are near the query location. Two broad categories of
such queries exist. The first encompasses queries that retrieve single
spatial web objects that each satisfy the query arguments. Most
proposals belong to this category. The second category, to which this
paper's proposal belongs, encompasses queries that support exploratory
user behavior and retrieve sets of objects that represent regions of
space that may be of interest to the user.
Specifically, the paper proposes a new type of query, namely the top-$k$
spatial textual clusters ($k$-STC) query that returns
the top-$k$ clusters that (i) are located the closest to a given
query location, (ii) contain the most relevant objects with regard to
given query keywords, and (iii) have an object density that exceeds a
given threshold.
To compute this query, we propose a basic algorithm that relies on
on-line density-based clustering and exploits an early stop
condition. To improve the response time, we design an advanced approach that includes three techniques: (i) an object skipping rule, (ii) spatially gridded posting lists, and (iii) a fast range query algorithm.
An empirical study on real data demonstrates that the paper's
proposals offer scalability and are capable of excellent performance.
\end{abstract}

%
%
%
%
%
%



\section{Introduction}
Spatial keyword query processing~\cite{DBLP:conf/sigmod/ChenSM06,DBLP:journals/pvldb/CongJW09,DBLP:conf/icde/FelipeHR08,DBLP:conf/ssdbm/HariharanHLM07,DBLP:conf/dexa/KhodaeiSL10,DBLP:journals/tkde/LiLZLLW11,DBLP:conf/cikm/ZhouXWGM05,DBLP:conf/ssd/RochaGJN11,DBLP:journals/tods/WuYJ13,DBLP:journals/vldb/WuCJ12,WuTKDE11,DBLP:conf/sigmod/LuLC11} allows users to submit queries that not only specify ``what'' the user is searching for, e.g., in the form of keywords or  names of consumer products of interest, but also specify ``where'' information, e.g.,  street addresses, postal codes, or geographic coordinates like latitude and longitude. For instance, a search may request a ``good micro-brewery that serves pizza'' and that is close to the user's hotel. In general, a spatial keyword query retrieves a set of spatial web objects that are located close to the query location and whose text descriptions are relevant to the query keywords. Figure~\ref{fig:topkpoint} illustrates an example spatial keyword query $q$ (located at the dot) with keywords `outdoor seating' that requests the top-5 restaurants (denoted by squares) in London from TripAdvisor\footnote{http://www.tripadvisor.com} based on a scoring function that performs a weighted sum on the spatial distance and text relevance of restaurants. 

Several variants of the spatial keyword query have been studied recently. They differ in terms of the query arguments and in how the relevance of the objects to the query arguments is determined (scoring functions). 
A continuously moving top-$k$ spatial keyword query~\cite{DBLP:conf/icde/WuYJC11,DBLP:journals/tods/WuYJ13} requests an up-to-date result while the query location changes continuously. 
A location-aware top-$k$ prestige-based text retrieval query~\cite{DBLP:journals/pvldb/CaoCJ10} increases the ranking of objects that have highly ranked nearby objects.
A collective spatial keyword query~\cite{DBLP:conf/sigmod/CaoCJO11} retrieves a set of objects that together match query keywords.
Rocha-Junior and N{\o}rv{\aa}g~\cite{RochaGJN12} consider spatial keyword search in road networks, and Li et al.~\cite{Li12} study the spatial keyword search constrained by the moving direction.

\begin{figure}[h]
	\centering
	\subfigure[Top-5 Objects]{\label{fig:topkpoint}\includegraphics[width=0.45\columnwidth]{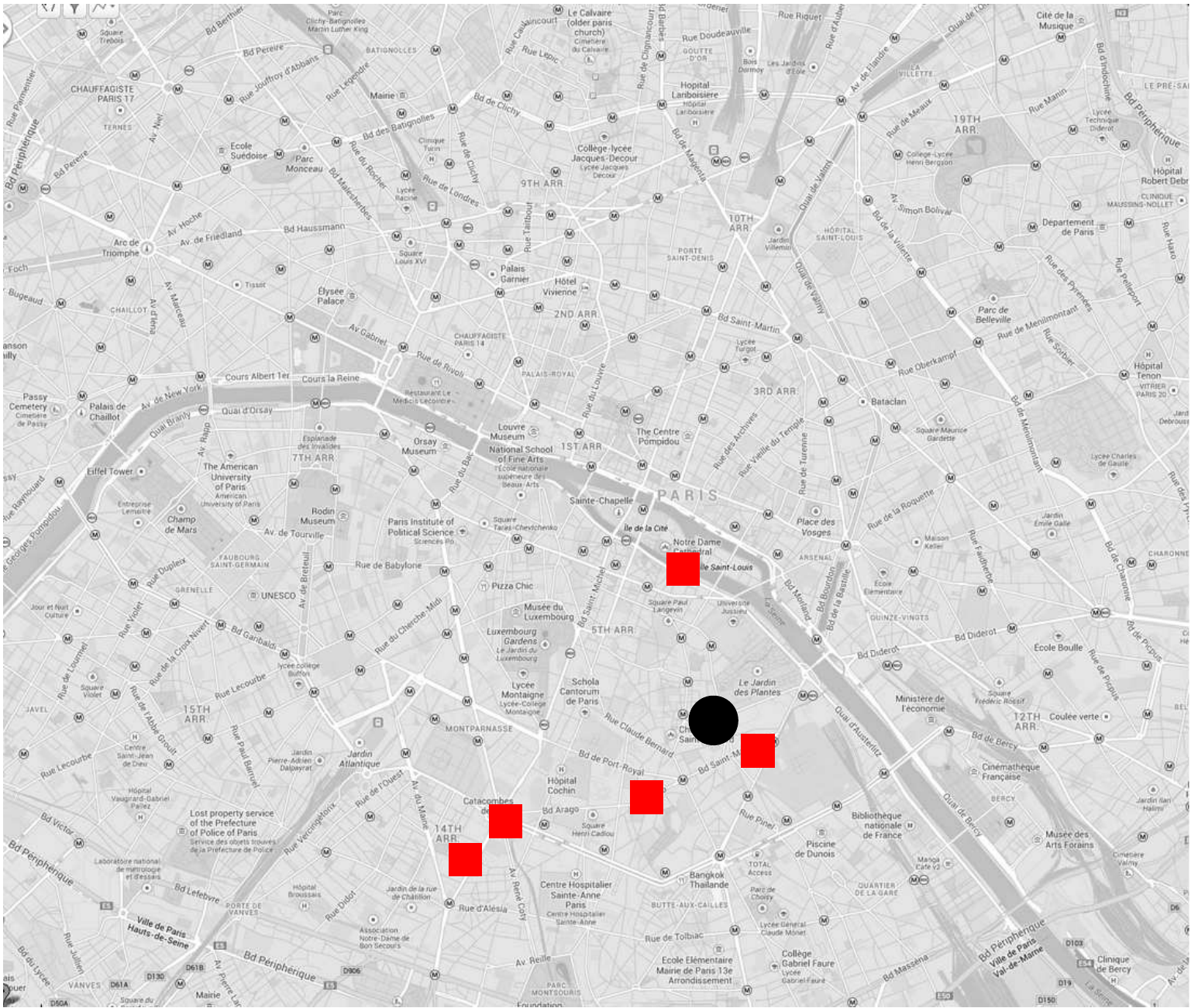}}\hspace{4ex}
	\subfigure[Top-5 Clusters]{\label{fig:topkregion}\includegraphics[width=0.45\columnwidth]{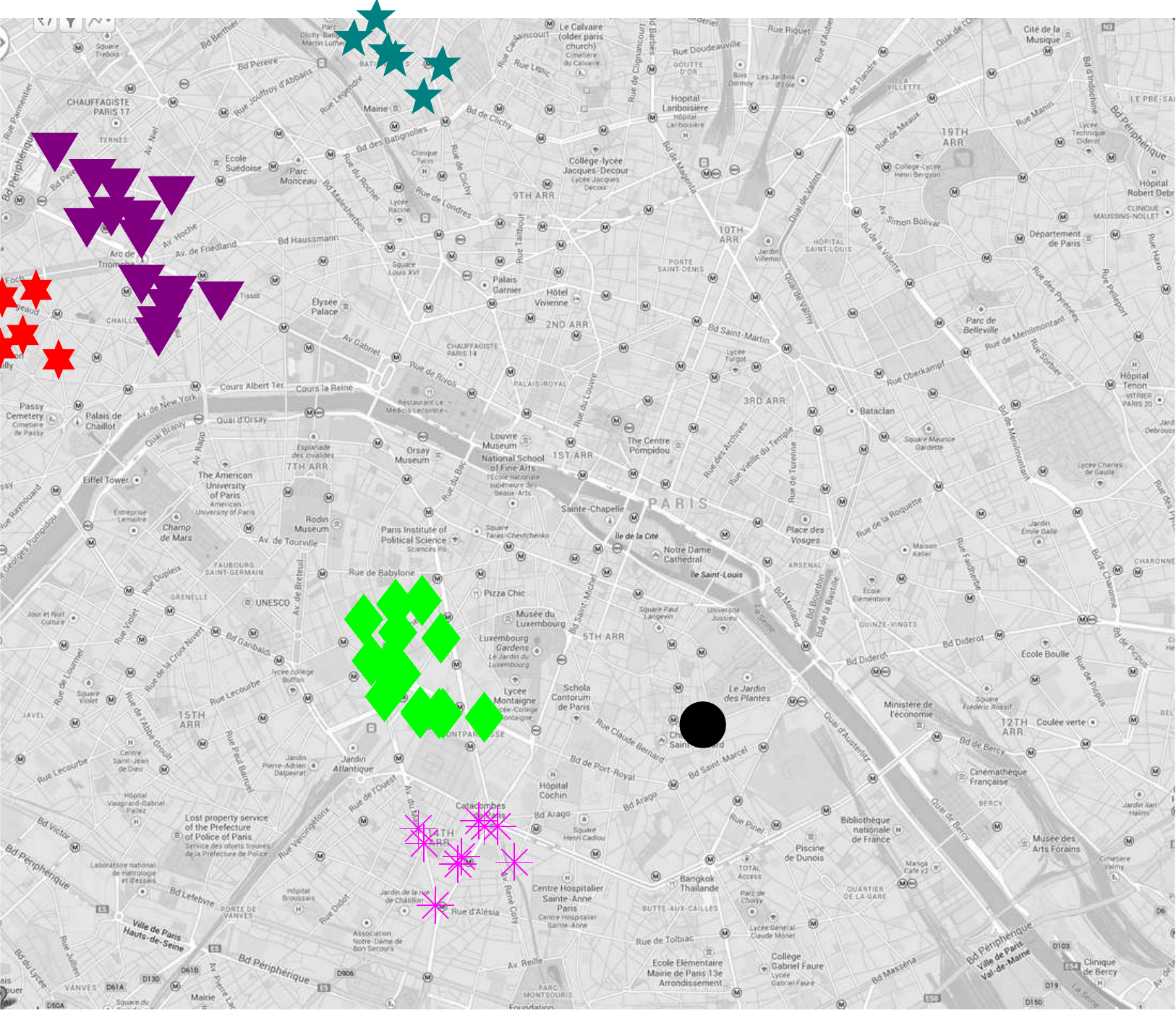}}
	\vspace{-2ex}
	\caption{Top-$k$ Objects vs. Top-$k$ Clusters}
\end{figure}

The majority of existing proposals retrieve single objects as elements in the result. However, users may be more interested in regions, or areas, that satisfy their query needs rather than in a set of objects that are  scattered in space. For instance, a user who is interested in purchasing jeans may prefer to visit a region with several shops that sell jeans, rather than visit shops that are further apart. In addition, similar objects (with the same functionality) are often located close to each other and thus form small regions, such as shopping, dinning, and entertainment districts~\cite{Durlauf}. Previous studies~\cite{DBLP:journals/pvldb/ChoiCT12,DBLP:journals/pvldb/TaoHCC13,DBLP:conf/cikm/LiuYS11} consider the co-location relationship between objects and retrieve regions so that the total weights of objects inside the result regions are maximized. However, these proposals impose specific shapes on the result regions, either fixed-size rectangles or circles. A recent study~\cite{DBLP:journals/pvldb/CaoCJY14} computes a maximum-sum region where the road network distance between objects is less than a query constraint and the sum of the scores of the objects inside the region is maximized. Nevertheless, this kind of query may retrieve a region containing many objects with low scores and may ignore a promising region with few objects with high scores. Also, the query includes a query range that reduces the search space. The performance is unknown if the query range is set to cover  the whole data set.

We aim at a solution that supports browsing, or exploratory, user behavior. It should have no constraints on the shapes of the retrieved regions. To this end, we view regions of interest as spatial textual clusters and propose and study a new type of query, namely the top-$k$ \textbf{Spatial Textual Cluster} ($k$-STC) query that returns the top-$k$ clusters such that (i) each cluster contains relevant spatial web objects with regard to query keywords, (ii) the density of each cluster satisfies a query constraint, and (iii) the clusters are ranked based on both their spatial distance and text relevance with regards to the query arguments. Figure~\ref{fig:topkregion} shows an example 5-STC query (black dot) with the same keywords (`outdoor seating') and location as the query in Figure~\ref{fig:topkpoint}. The top-5 spatial textual clusters (restaurants in London from TripAdvisor) are illustrated in the figure.

Clusters can be found using different approaches. We adopt density-based clustering that has several advantages: (i) no need to specify the number of clusters, (ii) clusters can have arbitrary shapes, and (iii) clusters are robust to outliers.
The two basic steps for top-$k$ STC retrieval are (i) obtaining the objects that are relevant to the query keywords and (ii) applying a density based clustering algorithm to these objects to find the top-$k$ clusters. The clustering algorithm processes the objects in a pre-defined order. It terminates when no  cluster with a better score can be found. We consider a cluster scoring function (introduced in Section~\ref{sec:prob-def}) that favors clusters close to the query location and that contain objects with high relevance with regard to the query keywords.
Different query keywords result in different clusters. Figure~\ref{fig:strex} demonstrates the top-5 clusters of restaurants obtained by two queries located at the same location (black dot), but with different keywords, namely `local cuisine' and `dessert'.
We are not aware of the query keywords until a query arrives, so pre-computing the clusters for all possible sets of query keywords is infeasible. We target an efficient algorithm that is able to find top-$k$ clusters with a response time that supports interactive search.

\begin{figure}[h]
	\begin{center}
		\subfigure[`local cuisine']{\label{fig:topkregion1}\includegraphics[width=0.45\columnwidth]{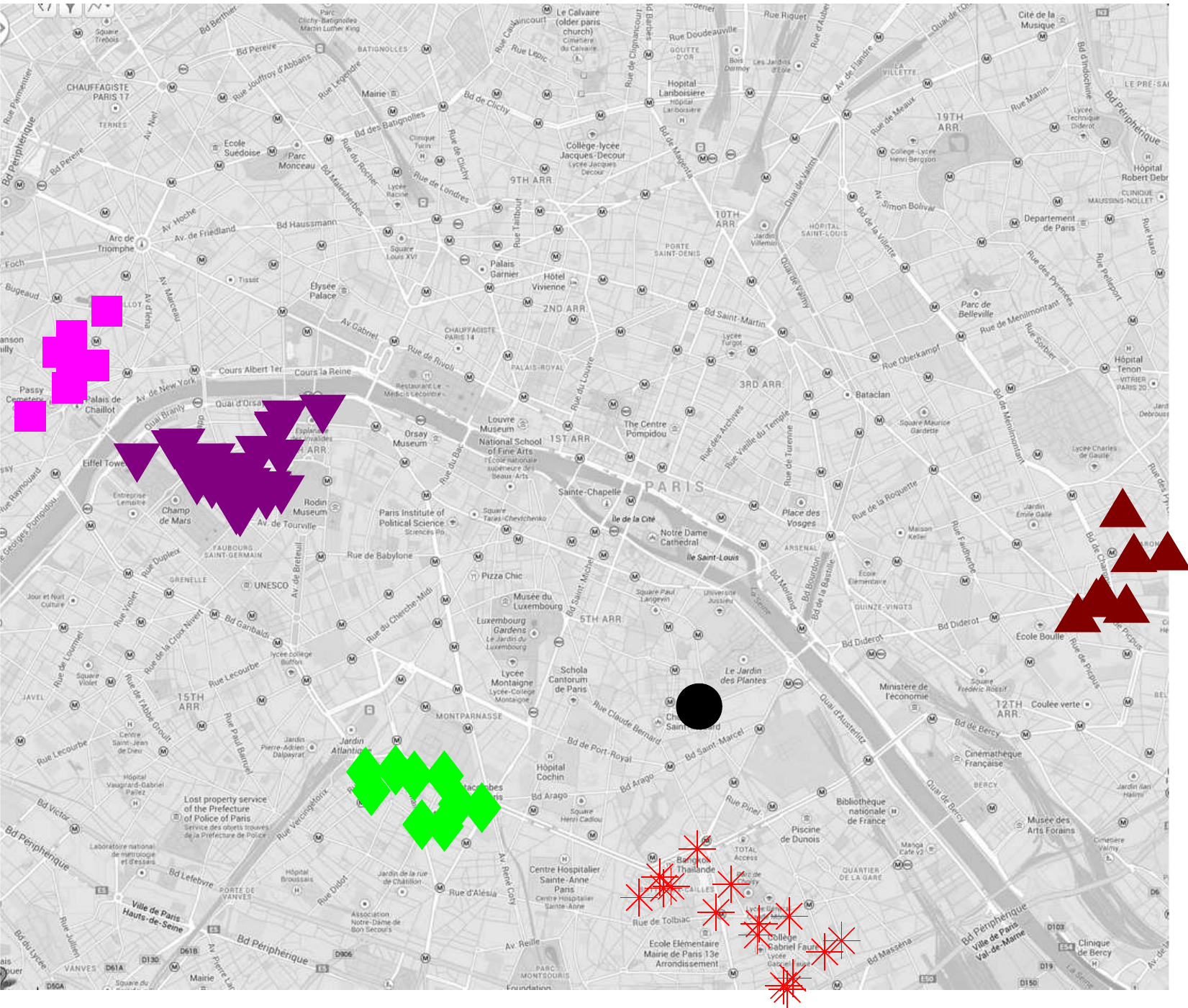}}\hspace{4ex}
		\subfigure[`dessert']{\label{fig:topkregion2}\includegraphics[width=0.45\columnwidth]{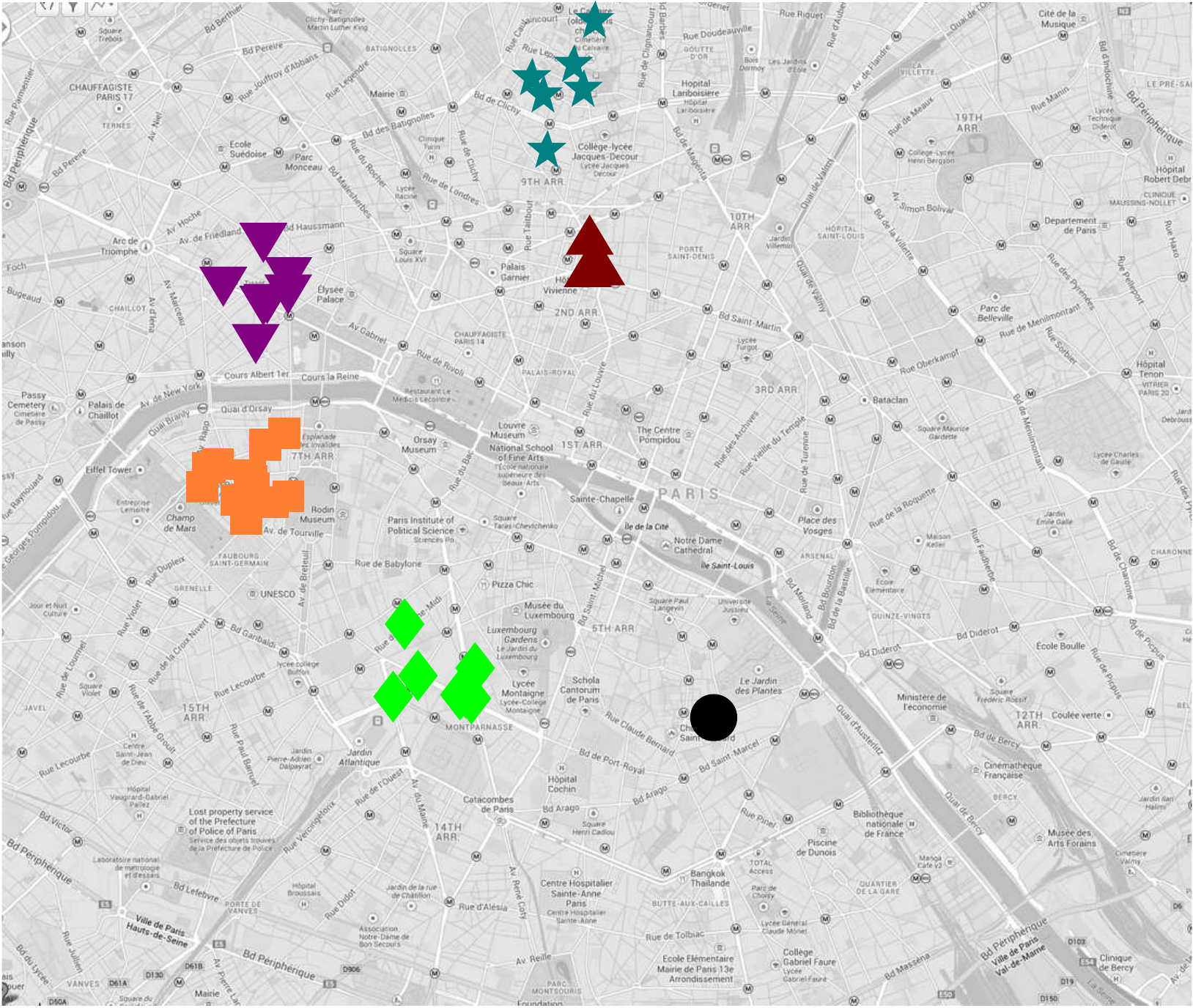}}
		\vspace{-2ex}
		\caption{Example Top-5 Clusters}
		\vspace{-2ex}
		\label{fig:strex}
	\end{center}
\end{figure}

More specifically, we propose a basic algorithm that combines on-line density-based clustering with an early stop condition. This algorithm applies the state-of-the-art density based clustering algorithm DBSCAN~\cite{DBLP:conf/kdd/EsterKSX96} to objects indexed by an  IR-tree~\cite{DBLP:journals/vldb/WuCJ12} to find top-$k$ clusters. This algorithm has to check the neighborhood of each relevant object in order to identify dense neighborhoods, since a cluster found by DBSCAN consists of  core objects and their dense neighborhoods. We propose an advance algorithm that reduces the number of objects to be examined. Moreover, determining whether a neighborhood is dense or not is time-consuming, since it involves range queries on the IR-tree. We design spatially gridded posting lists (SGPL) to estimate the selectivity of range queries on the IR-tree, so that sparse neighborhoods can be detected quickly without querying the index, thus saving computational cost. Further, SGPL is able to handle the  necessary range queries more efficiently compared with the basic algorithm.


To the best of our knowledge, we are the first to study the top-$k$ spatial textual clusters query. This paper's contributions are: 
\begin{itemize}\setlength{\itemsep}{-1pt}
	\item We introduce the top-$k$ spatial textual clusters query that returns close, relevant, and dense clusters in order to support exploratory user behavior.
	\item A basic algorithm that exploits the density-based algorithm DBSCAN with an early stop condition on the IR-tree.
	\item In order to provide response times that enable interactive search, we propose an advanced approach that includes the following techniques.
	\begin{itemize}\setlength{\itemsep}{-1pt}
		\item In order to retrieve clusters, the neighborhood of each object has to be checked in the basic algorithm. We design a skipping rule that reduces the number of objects to be examined.
		\item A result cluster consists of dense neighborhoods. However, determining whether a neighborhood is dense in the basic algorithm involves an expansive range query on the index. We design spatially gridded posting lists (SGPL) to estimate the selectivity of range queries, thus enabling the pruning of sparse neighborhoods.
		\item The SGPL is also able to support range queries, so that a performance gain is obtained for the range queries that cannot be avoided, compared with the basic algorithm.
	\end{itemize}
	\item An extensive empirical study with real data demonstrates that the paper's proposals offer scalability and are capable of excellent performance. 
\end{itemize}

The rest of the paper is organized as follows. Section~\ref{sec:prob-def} formally defines the top-$k$ spatial textual cluster query. The basic algorithm is presented in Section~\ref{sec:baseline}. We present the advanced approach in Section~\ref{sec:advance}, including three techniques: (i) a skipping rule that reduces the number of objects to be examined (Section~\ref{sec:reduce}), (ii) the spatially gridded posting lists (SGPL) for selectivity estimation of range queries (Section~\ref{sec:histogram}), and (iii) a fast range query processing algorithm on SGPL (Section~\ref{sec:fastrange}). 
We report on the empirical performance study in Section~\ref{sec:exp}. Finally, we cover related work in Section~\ref{sec:related} and offer conclusions and research directions in Section~\ref{sec:con}.

\section{Problem Definition}\label{sec:prob-def}
We consider a data set $\mathcal{D}$ in which each object $p \in
\mathcal{D}$ is a pair $\langle \lambda,\psi \rangle$ of a point
location $p.\lambda$ and a text description, or document, $p.\psi$
(e.g., the facilities and menu of a restaurant). Document $p.\psi$
is represented by a vector $(w_1,w_2, \cdots , w_i)$ in which each
dimension corresponds to a distinct term $t_i$ in the document. The
weight $w_i$ of a term in the vector can be computed in several
different ways, e.g., using tf-idf weighting~\cite{Salton1975} or
language models~\cite{Ponte1998}.

We adopt the density-based clustering model~\cite{DBLP:conf/kdd/EsterKSX96}, and clusters are query-dependent. We proceed to present the relevant definitions extended for the top-$k$ spatial textual cluster query.

\begin{definition}
	Given a set of keywords $\psi$, the \textbf{relevant object set} $D_\psi$ satisfies (i) $D_\psi \subseteq \mathcal{D}$ and (ii) $\forall p \in D_\psi (\psi \cap p.\psi \neq \emptyset)$. 
\end{definition}

\begin{definition}
	The \textbf{$\epsilon$-neighborhood} of  a relevant object $p \in D_\psi$, denoted by $N_{\epsilon}(p)$, is defined as $N_{\epsilon}(p)=\{p_i \in D_\psi~|~\eucdist{p}{p_i} \leq \epsilon\}$. 
\end{definition}

\begin{definition}
	A \textbf{dense $\epsilon$-neighborhood} of a relevant object $N_{\epsilon}(p)$ contains at least $\mathit{minpts}$ objects, i.e., $|N_{\epsilon}(p)| \geq \mathit{minpts}$. 
\end{definition}

\begin{definition}
	A relevant object $p$ is a \textbf{core} if its $\epsilon$-neighborhood is dense.
\end{definition}

\begin{definition}
	A relevant object $p_i$ is \textbf{directly reachable} from a relevant object $p_j$ with regard to $\epsilon$ and $\mathit{minpts}$ if
	\begin{enumerate}\setlength{\itemsep}{-1pt}
		\item $p_i \in N_{\epsilon}(p_j)$ and
		\item $|N_{\epsilon}(p_j)| \geq \mathit{minpts}$.
	\end{enumerate}
\end{definition}

\begin{definition}
	A relevant object $p_i$ is \textbf{reachable} from a relevant object $p_j$ with regard to $\epsilon$ and $\mathit{minpts}$ if there is a chain of relevant objects $p_1,\cdots,p_n$, where $p_i=p_1$, $p_j=p_n$, such that $p_{m}$ is directly reachable from $p_{m+1}$ for $1 \leq m < n$. 
\end{definition}

\begin{definition}
	A relevant object $p_i$ is \textbf{connected} to a relevant object $p_j$ with regard to $\epsilon$ and $\mathit{minpts}$ if there is a relevant object $p_m$ such that both $p_i$ and $p_j$ are reachable from $p_m$ with regard to $\epsilon$ and $\mathit{minpts}$.
\end{definition}

\begin{definition}
	A \textbf{spatial textual cluster} $R$ with regard to $\psi, \epsilon$, and $\mathit{minpts}$  satisfies the following condition: 
	\begin{enumerate}\setlength{\itemsep}{-1pt}
		\item $R \subseteq D_\psi$ and
		\item $R$ is a maximal set such that $\forall p_i, p_j \in R$, $p_i$ and $p_j$ are connected through dense $\epsilon$-neighborhoods when considering only objects in $D_\psi$.
	\end{enumerate}
\end{definition}

A spatial textual cluster is a density-based cluster~\cite{DBLP:conf/kdd/EsterKSX96} found from the relevant object set $D_\psi$ with regard to the query keywords $\psi$. 
%
%
%
A \textbf{top-$k$ Spatial Textual Cluster} ($k$-STC) query
$q=\langle \lambda, \psi, k, \epsilon,$ $\mathit{minpts} \rangle$ takes five 
arguments: a point location $\lambda $, a set of keywords $\psi$, a
number of requested object sets $k$, a distance constraint $\epsilon$ on neighborhoods, and the minimum number of objects $\mathit{minpts}$ in a dense $\epsilon$-neighborhood.
It returns
a list of $k$ spatial textual clusters that minimize a scoring
function and that are in ascending order
of their scores. The maximality of each cluster implies that the top-$k$ clusters do not overlap.
The density requirement parameters $\epsilon$ and $\mathit{minpts}$ are able to capture how far the user is willing to move before reaching another place of interest. They depend on the users' preferences.

Intuitively, a cluster with high text relevance and that is located close to the query location should be given a high ranking in the result. We thus use the following scoring function. 
\begin{equation}\label{equ:rankf}
	\mathit{score}_q(R) = \alpha \cdot d_{q.\lambda}(R) + (1-\alpha) \cdot (1-\mathit{tr}_{q.\psi}(R)),
\end{equation}
where $d_{q.\lambda}(R)$ is the minimum spatial distance between the query location and the objects in $R$ and $\mathit{tr}_{q.\psi}(R)$ is the maximum text relevance in $R$. The approaches we present are applicable to scoring functions that are monotone with respect to both spatial distance and text relevance. Parameter $\alpha$ is used to balance the spatial proximity and the text relevance of the retrieved clusters. Note that all spatial distances and text relevances are normalized to the range $[0,1]$.

\begin{example}\label{ex:pex}
	Consider the example $k$-STC query $q$ with location $q.\lambda$ and $q.\epsilon$ as shown in Figure~\ref{fig:pex}(a) and with $q.\psi=\{\mathit{coffee, tea}\}$, $q.k=1$, and $q.\mathit{minpts}=2$. The data set contains the 7 objects $p_1,p_2, \cdots$, $p_7$ shown in Figure~\ref{fig:pex}(a). Figure~\ref{fig:pex}(b) shows the document vector and the Euclidean distance to the query location of each object. Let $\alpha=0.5$ and $\mathit{tr}_{q.\psi}(p.\psi)=\sum_{t \in q.\psi \cap p.\psi}w_t$. The top-1 cluster is $R=\{p_3,p_5\}$ that has score 0.315 ($=0.5 \times (0.11+0.15)/2 + (1-0.5) \times (1-(0.5+0.5)/2)$).
\end{example}

\begin{figure}[h]
	\vspace{-4ex}
	\begin{center}
		\small
		\begin{tabular}{@{}c@{}@{}c@{}}
			\begin{tabular}{c}
				\includegraphics[width=0.3\columnwidth]{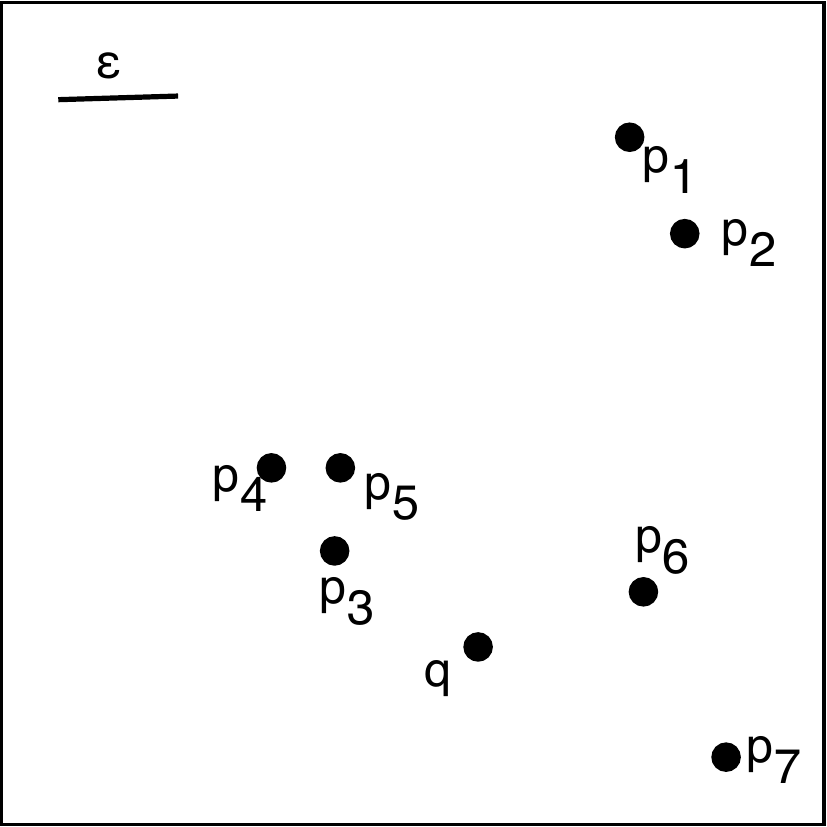}
			\end{tabular}
			&  
			\begin{tabular}{|cccc||c|}
				\hline
				& coffee & tea & pizza & $d_{q.\lambda}(p_i.\lambda)$ \\
				\hline
				$p_1$ &(0.2, & 0.2,&0)&0.25\\
				\hline
				$p_2$ &(0.2, & 0.2,&0)&0.2\\
				\hline
				$p_3$ & (0.5,&0,&0)&0.11\\
				\hline
				$p_4$ & (0,&0,&0.5)&0.18\\
				\hline
				$p_5$ & (0,&0.5,&0)&0.15\\
				\hline
				$p_6$ & (0.5,&0,&0)&0.1\\
				\hline
				$p_7$ & (0.5,&0.5,&0)&0.19\\
				\hline
			\end{tabular} \\
			(a)&(b)
		\end{tabular}
		\caption{Example $k$-STC Query}
		\label{fig:pex}
		\vspace{-4ex}
	\end{center}
\end{figure}

\section{Basic Approach}\label{sec:baseline}
We first briefly introduce the
index structures used for organizing objects and then present the basic algorithm for the processing of $k$-STC queries.
\subsection{Indexes}

We adopt the IR-tree~\cite{DBLP:journals/vldb/WuCJ12} and inverted file~\cite{Zobel06} index structures to organize objects.

An inverted file index has two main components.
\begin{itemize}\setlength{\itemsep}{-1pt}
	\item A vocabulary of all distinct words appearing in the text
	descriptions of the objects in the data set.
	\item A posting list for each word $t$, i.e., a sequence of pairs $(\mathit{id},w)$, where $\mathit{id}$ is the
	identifier of an object whose text description contains $t$ and $w$ is the word's weight in the object.
\end{itemize}

The IR-tree is an R-tree~\cite{Guttman84} extended with inverted files. Each leaf node contains entries of the form $e_0=(\mathit{id},\Lambda)$, where $e_0.\mathit{id}$ refers to an object identifier and $e_0.\Lambda$ is a minimum bounding rectangle (MBR) of the spatial location of the object. Each leaf node also contains a pointer to an inverted file indexing the text descriptions of all objects stored in the node. Each non-leaf node $N$  contains entries of the form $e=(\mathit{id},\Lambda)$, where $e.\mathit{id}$ points to a child node of $N$ and $e.\Lambda$ is the MBR of all rectangles in entries of the child node. Each non-leaf node also contains a pointer to an inverted file indexing the pseudo text descriptions of the entries stored in the node. A pseudo text description of an entry $e$ is a summary of all (pseudo) text descriptions in the entries of the child node pointed to by $e$. This enables the derivation of an upper bound on the text relevance to a query of any object contained in the subtree rooted at $e$.

\begin{example}
	Table~\ref{tab:invertedfile} shows the inverted file indexing the 7 objects in Figure~\ref{fig:pex}. For example, the posting list for word `pizza' tells that the text description of $p_4$ contains `pizza' and the weight is 0.5. Figure~\ref{fig:irtree} illustrates the IR-tree with fanout 2 indexing the 7 objects in Figure~\ref{fig:pex}. As a specific example, the weight of
	`coffee' for entry $N_6$ in inverted file $\mathit{IF}_7$ is 0.5, which is the maximal
	weight of `coffee' in the two documents in the child node of $N_6$.
\end{example}
\vspace{-4ex}
\begin{table}[h]
	\caption{Example Inverted File}\label{tab:invertedfile}
	\begin{center}
		\begin{tabular}{|l|l|}
			\hline
			coffee & $(p_3,0.5),(p_6,0.5),(p_7,0.5),(p_1,0.2),(p_2,0.2)$ \\
			\hline
			tea & $(p_5,0.5),(p_7,0.5),(p_1,0.2),(p_2,0.2)$\\
			\hline
			pizza & $(p_4,0.5)$\\
			\hline
		\end{tabular}
	\end{center}
\end{table}%

\begin{figure}[h]
	\vspace{-4ex}
	\begin{center}
		\includegraphics[width=\columnwidth]{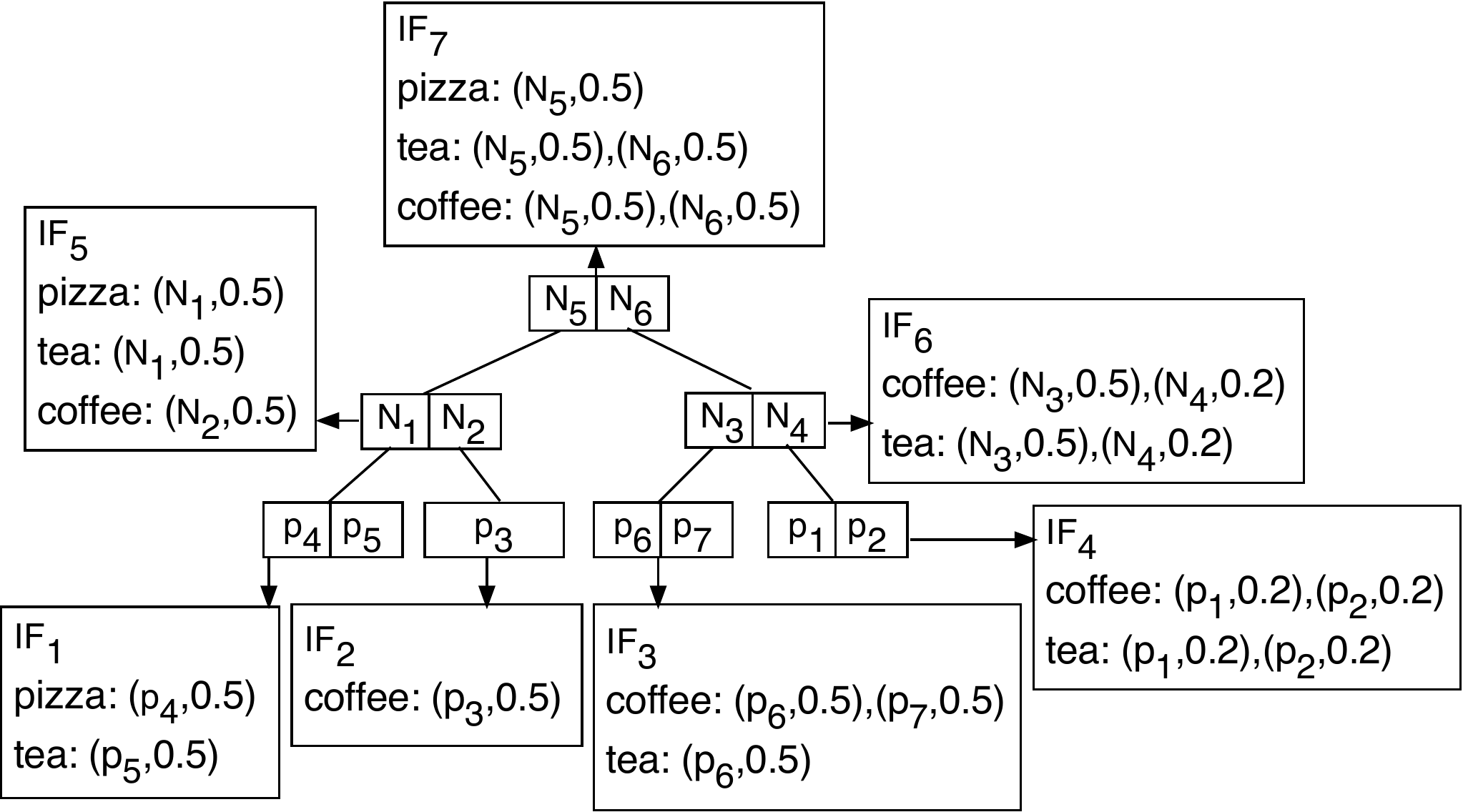}
		\caption{Example IR-tree}
		\label{fig:irtree}
		\vspace{-4ex}
	\end{center}
\end{figure}

\subsection{Algorithm}
Given a query, the top-$k$ spatial textual clusters are the top-$k$ density-based clusters found from the relevant object set $D_\psi$ with regard to the query keywords. A straightforward solution is to first obtain the relevant object set $D_\psi$ and then to find all density-based clusters in $D_\psi$. These clusters are then sorted according to the scoring function (Equation~\ref{equ:rankf}), and the top-$k$ best clusters are returned as the result. This straightforward solution is inefficient, since finding all clusters is expensive. The proposed basic algorithm is able to return the top-$k$ clusters without first finding all clusters. Specifically, some candidate clusters are first obtained. A threshold is set according to the  score of the $k$-th candidate cluster. The basic algorithm estimates a bound on the scores of all unfound clusters. If the bound is worse than the threshold, the currently found top-$k$ clusters are the result. 


Algorithm~\ref{alg:basic} shows the pseudo code of the basic algorithm. 
It first obtains the relevant object set $D_\psi$ with regard to the query keywords by taking the union of the posting lists of the query terms in the inverted index (line 1). 
Next, it sorts the objects in $D_\psi$ in ascending order of their Euclidean distances to the query location, i.e., $d_{q.\lambda}(p.\lambda)$, and keeps the sorted copy in
$\mathit{slist}$ (line 2).
In addition, the objects in $D_\psi$ are sorted  in ascending order of their converted text relevance with regard to the query keywords, i.e., $1-\mathit{tr}_{q.\psi}(p.\psi)$, and the sorted copy is kept in
$\mathit{tlist}$ (line 3).

 The candidate list $\mathit{rlist}$ is initialized as empty (line 4). The threshold is set to infinity (line 5).
The algorithm does sorted access in parallel to each of the two sorted lists $\mathit{slist}$ and $\mathit{tlist}$ (line 7). Specifically, the algorithm accesses the top member of each of the lists under sorted access, then accesses the second member of each of the lists, and so on. As an object $p$ is obtained from one of the lists, function $\mathbf{GetCluster}$ (Algorithm~\ref{alg:expand}) tries to retrieve the cluster containing $p$ as a core object (line 8). Meanwhile, the objects contained in the cluster are removed from both $\mathit{slist}$ and $\mathit{tlist}$. If the retrieved cluster is not empty, it is added to the candidate list $\mathit{rlist}$, and the threshold $\tau$ is updated to the score of the $k$-th candidate in $\mathit{rlist}$ (lines 9--11). The algorithm estimates a lower bound on the  scores of all unfound clusters by using the minimal Euclidean distance $\mathit{sb}$ and minimal converted text relevance $\mathit{tb}$ of all objects left in $\mathit{slist}$ and $\mathit{tlist}$, i.e., $\mathit{bound}=\alpha \cdot \mathit{sb} + (1-\alpha) \cdot \mathit{tb}$ (lines 12--14). The result is guaranteed to be found when $\mathit{bound} \geq \tau$ or $\mathit{slist}$ is exhausted (indicating that $\mathit{tlist}$ is also exhausted) (line 15). The top-$k$ candidate clusters in $\mathit{rlist}$ are finalized as the result.

\begin{lemma}\label{lem:earlystop}
	The stop condition (line 15) in Algorithm~\ref{alg:basic} guarantees the correct result.
\end{lemma}
\begin{proof}
	The stop condition is satisfied when either of the following two conditions is true: (i) $\mathit{slist} = \emptyset$ or (ii) $\mathit{bound} \geq \tau$. Condition (i) means that all objects have been clustered or identified as noise. Hence, the top-$k$ candidate clusters are the final result, since no further clusters can be found. Consider condition (ii). For any unfound cluster $R$, 
	\begin{eqnarray}
		\mathit{score}_q(R) &=& \alpha \cdot d_{q.\lambda}(R) + (1-\alpha) \cdot (1-\mathit{tr}_{q.\psi}(R))\nonumber \\
		&\geq& \alpha \cdot sb + (1-\alpha) \cdot tb\nonumber \\
		&=& \mathit{bound} \geq \tau. \nonumber 
	\end{eqnarray}
	Hence, no cluster can have a better (smaller) score than does the top-$k$ candidates. 
\end{proof}

\begin{algorithm}[h]
	\caption{$\mathbf{Basic}$(Query $q$, IR-tree $\mathit{irtree}$, InvertedIndex $\mathit{iindex}$, Integer $k$)}\label{alg:basic}
	\begin{algorithmic}[1]
		\State $D_\psi \leftarrow \mathbf{LoadRelevantObjects}(q.\psi,iindex)$;
		\State $\mathit{slist \leftarrow}$ sort objects in $D_\psi$ in ascending order of $d_{q.\lambda}(p.\lambda)$;
		\State $\mathit{tlist \leftarrow}$ sort objects in $D_\psi$ in ascending order of $1-\mathit{tr}_{q.\psi}(p.\psi)$;
		\State $\mathit{rlist} \leftarrow \emptyset$;
		\State $\tau \leftarrow \infty$;
		\Repeat
		\State Object $p \leftarrow$ sorted access in parallel to $\mathit{slist}$ and $\mathit{tlist}$;
		\State $c \leftarrow \mathbf{GetCluster}(p,q,\mathit{irtree},\mathit{tlist}, \mathit{slist})$;
		\If{$c \neq \emptyset$}
		\State Add $c$ to $\mathit{rlist}$;
		\State $\tau \leftarrow$ score of the $k$-th cluster in $\mathit{rlist}$;
		\EndIf
		\State $\mathit{sb} \leftarrow \mathbf{First}(\mathit{slist})$;
		\State $\mathit{tb} \leftarrow \mathbf{First}(\mathit{tlist})$;
		\State $\mathit{bound} \leftarrow \alpha \cdot \mathit{sb} + (1-\alpha) \cdot \mathit{tb}$;
		\Until{$\mathit{bound} \geq \tau \vee \mathit{slist} = \emptyset$}
		\State Return $\mathit{rlist}$;
	\end{algorithmic}
\end{algorithm}

To retrieve a cluster $R$ containing $p$ as a core object, function $\mathbf{GetCluster}$ (see Algorithm~\ref{alg:expand}) issues a range query centered at $p$ with radius $q.\epsilon$ on the IR-tree (line 2). The goal is to check whether the $\epsilon$-neighborhood of $p$ is dense. If the result set $\mathit{neighbors}$ of the range query contains fewer than $q.\mathit{minpts}$ objects (a sparse neighborhood), object $p$ is marked as noise, and an empty set is returned (lines 3--6). Otherwise, $\mathit{neighbors}$ is considered as a temporary cluster (line 8). Next, the temporary cluster is expanded by checking the $\epsilon$-neighborhood of each object $p_i$ in $\mathit{neighbors}$ except $p$ (lines 12 and 13). If the $\epsilon$-neighborhood of $p_i$ is dense (line 14), the objects inside and previously labeled as noise are added to the temporary cluster (lines 16 and 17). Next, the objects inside and not belonging to the temporary cluster are added to both the temporary cluster and to $\mathit{neighbors}$ in preparation for further expansion in subsequent iterations (lines 18--21). In order to avoid duplicate operations, the objects that are marked as noise or are added to the temporary cluster are removed from lists $\mathit{tlist}$ and $\mathit{slist}$ (lines 4, 9, and 20). The temporary cluster $R$ is finalized and returned if no more object can be added.

Function $\mathbf{RangeQuery}$ (Algorithm~\ref{alg:range}) is used to find  the objects in the $\epsilon$-neighborhood of an object $p$ using an IR-tree on the objects. A priority queue organizes the nodes to be visited in the IR-tree, using the minimum Euclidean distance between nodes and the query as the key. The queue is initialized as empty (line 2). First, the root node of the IR-tree is added to the queue (line 3). The algorithm iteratively visits and removes the first node in the queue (lines 5 and 6). If the node is a leaf node, the objects in it that are relevant to the query keywords and that are located in range $q.\epsilon$ are added to the result $\mathit{neighbors}$ (lines 7--10). Otherwise, the node is a non-leaf node. Its child nodes that are relevant to the query keywords and that are located in range $q.\epsilon$ are inserted into the priority queue (lines 12--14). The process terminates when the queue is exhausted.

\begin{algorithm}[h]
	\caption{$\mathbf{GetCluster}$(Object $p$, Query $q$, IR-tree $\mathit{irtree}$, List $\mathit{tlist}$, List $\mathit{slist}$)}\label{alg:expand}
	\begin{algorithmic}[1]
		\State $R \leftarrow \emptyset$;
		\State $\mathit{neighbors} \leftarrow \mathit{irtree}.\mathbf{RangeQuery}(q,p)$;
		\If{$\mathit{neighbors.size} < q.\mathit{minpts}$} 
		\State Remove $p$ from $\mathit{tlist}$ and $\mathit{slist}$;
		\State Mark $p$ as a noise;
		\State Return $R$;
		\Else  \Comment{$p$ is a core;}
		\State Add $\mathit{neighbors}$ to $R$;
		\State Remove $\mathit{neighbors}$ from $\mathit{tlist}$ and $\mathit{slist}$;
		\State Remove $p$ from $\mathit{neighbors}$;
		\While{$\mathit{neighbors}$ is not empty}
		\State Object $p_i \leftarrow$ remove an object from $\mathit{neighbors}$;
		\State $\mathit{neighbors}_i \leftarrow \mathit{irtree}.\mathbf{RangeQuery}(q,p_i)$;
		\If{$\mathit{neighbors}_i.\mathit{size} \geq q.\mathit{minpts}$}
		\For{each object $p_j$ in $\mathit{neighbors}_i$}
		\If{$p_j$ is a noise}
		\State Add $p_j$ to $R$;
		\ElsIf{$p_j \notin R$}
		\State Add $p_j$ to $R$;
		\State Remove $p_j$ from $\mathit{tlist}$ and $\mathit{slist}$;
		\State Add $p_j$ to $\mathit{neighbors}$;
		\EndIf
		\EndFor
		\EndIf
		\EndWhile
		\State Return $R$;
		\EndIf
	\end{algorithmic}
\end{algorithm}

\begin{algorithm}[h]
	\caption{$\mathbf{RangeQuery}$(Query $q$, Object $p$)}\label{alg:range}
	\begin{algorithmic}[1]
		\State $\mathit{neighbors \leftarrow \emptyset}$;
		\State PriorityQueue $\mathit{queue} \leftarrow \emptyset$;
		\State $\mathit{queue}.\mathbf{Enqueue}(\mathit{root})$;
		\While{$\mathit{queue}$ is not empty}
		\State $e \leftarrow \mathit{queue}.\mathbf{Dequeue}()$;
		\State $N \leftarrow \mathbf{ReadNode(e)}$;
		\If{$N$ is a leaf node}
		\For{each object $o$ in $N$}
		\If{$o$ is relevant to $q.\psi$ and $\eucdistmin{p}{o} \leq q.\epsilon$}
		\State $\mathit{neighbors}.\mathbf{Add}(o)$;
		\EndIf
		\EndFor
		
		\Else
		\For{each entry $e'$ in $N$}
		\If{$e'$ is relevant to $q.\psi$ and $\eucdistmin{p}{e'} \leq q.\epsilon$}
		\State $\mathit{queue}.\mathbf{Enqueue}(e')$;
		\EndIf
		\EndFor
		\EndIf	
		\EndWhile
		\State Return $\mathit{neighbors}$;
	\end{algorithmic}
\end{algorithm}

\begin{example}
	Consider the $k$-STC query in Example~\ref{ex:pex}. Given $q.\psi=\{\mathit{coffee, tea}\}$, the basic algorithm first retrieves $D_\psi=\{p_1,p_2,p_3,p_5,p_6,p_7\}$ using the inverted file shown in Table~\ref{tab:invertedfile}. Then, the two sorted lists are calculated, i.e., $\mathit{slist}=((p_6,0.1),(p_3,$ $0.11),(p_5,0.15),(p_7,0.19)$, $(p_2,0.2),(p_1,0.25))$ and $\mathit{tlist}=((p_7,\\0),(p_3,$ $0.5),(p_5,0.5),(p_6,$ $0.5),(p_1,0.6),(p_2,$ $0.6))$. By accessing the two sorted lists in parallel, $p_7$ is obtained first and is used to obtain a cluster. However, the neighborhood of $p_7$ is sparse, so $p_7$ is marked as noise and is removed from lists $\mathit{slist}$ and $\mathit{tlist}$. Next, $p_6$ is also marked as noise and is removed from  $\mathit{slist}$ and $\mathit{tlist}$. The algorithm proceeds to consider $p_3$ and obtains a cluster $R=\{p_3,p_5\}$. Then $\tau$ is set to the score of $R$, i.e., 0.315. The two sorted lists become $\mathit{slist}=((p_2,0.2),(p_1,0.25))$ and $\mathit{tlist}=((p_1,0.6),(p_2,0.6))$. The bound is calculated as  $\mathit{bound} = 0.5 \times 0.2 + 0.5 \times 0.6=0.4$. Thus, we get $\mathit{bound} \geq \tau$. The algorithm can then safely return $R$ as the top-1 result without processing $p_1$ and $p_2$.
\end{example}

\section{Advanced Approach}\label{sec:advance}
The basic approach is inefficient due to two main reasons.
\begin{itemize}\setlength{\itemsep}{-1pt}
	\item It checks the neighborhoods of all relevant objects with regard to the query keywords, which is expensive.
	\item Checking the neighborhood of an object involves a time-consuming range query on the index.
\end{itemize}
The advanced approach includes three techniques that eliminate the above disadvantages and achieve improved performance.

\subsection{Object Skipping}\label{sec:reduce}
Function $\mathbf{GetCluster}$ tries to find a cluster $R$ containing a relevant object $p$ as a core object. It first determines whether the $\epsilon$-neighborhood of $p$ is dense. If so, cluster $R$ is initialized as the set of relevant objects inside the $\epsilon$-neighborhood of $p$. We consider an object as examined if its neighborhood has been checked. Next, the relevant objects other than $p$ inside the $\epsilon$-neighborhood of $p$ are examined one by one. If a neighborhood under consideration is dense, the newly found relevant objects inside it are added to cluster $R$. This way, cluster $R$ is finalized when each relevant object has been examined. However, it is possible to get cluster $R$ by examining only a portion of the relevant objects. Consider the example in Figure~\ref{fig:reduce}. The neighborhoods of the four objects $p_1$, $p_2$, $p_3$, and $p_4$ are illustrated by dashed and solid circles. It can be seen that the neighborhood of $p_4$ (solid circle) is covered by the union of the neighborhoods of $p_1$, $p_2$, and $p_3$ (dashed circles). In this case, checking the neighborhood of $p_4$ is unnecessary after having examined $p_1$, $p_2$, and $p_3$, since the objects inside the neighborhood of $p_4$ are guaranteed to have already been considered. In other words, examining $p_4$ cannot contribute more objects to the cluster. Based on this observation, we define a skipping rule that reduces the number of relevant objects to be examined, and we design an algorithm $\mathbf{OS}$ that implements the rule.

\begin{figure}[h]
	\begin{center}
		\includegraphics[width=0.3\columnwidth]{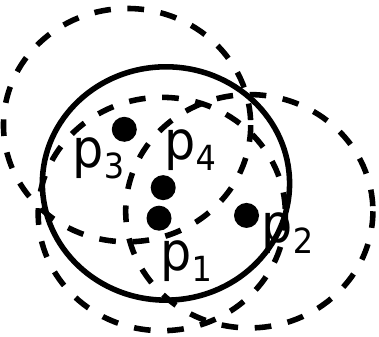}
		\caption{Neighborhoods of Objects}
		\label{fig:reduce}
		\vspace{-4ex}
	\end{center}
\end{figure}

\begin{skippingrule}
	Let $S=(p_1,p_2,\cdots,p_n)$ be the order in which a set of objects is examined. Object $p_i$ $(i > 1)$ can be skipped if the neighborhood of $p_i$ is fully covered by the union of the neighborhoods of the objects examined before $p_i$, i.e., if $N_\epsilon(p_i) \subset \cup_{1 \leq j < i}N_\epsilon(p_j)$, where $N_\epsilon(p_i)$ is represented as a circular region centered at $p_i$ and with radius $\epsilon$.
\end{skippingrule}

The skipping rule is effective when given a good ordering $S$. Consider the example in Figure~\ref{fig:reduce}. If $S=(p_1,p_2,p_4,p_3)$, no object can be skipped. However, if $S=(p_1,p_2,p_3,p_4)$, object $p_4$ can be skipped. Intuitively, if the union of the neighborhoods of the objects that have been examined covers a large area, the probability of skipping the next object is high. 
We propose an algorithm $\mathbf{OS}$ that implements the skipping rule. It follows Algorithm~\ref{alg:expand} with the following differences.
\begin{itemize}\setlength{\itemsep}{-1pt}
	\item Given an object $p$ and its neighborhood, the objects inside the neighborhood are sorted in descending order of their distance to $p$. The motivation is that the farther the objects are from $p$, the larger the area covered by their neighborhoods is. Referring to lines 2 and 13 in Algorithm~\ref{alg:expand}, the objects returned from $\mathbf{RangeQuery}$ are sorted in descending order of their distances. Let $S(p)$ be the sorted list of the objects inside the neighborhood of $p$. Note that list $\mathit{neighbors}$ (line 2 in Algorithm~\ref{alg:expand}) maintains the objects to be examined. It is initialized as the sorted list $S(p)$. For each object $p_i$ in $S(p)$, the sorted list $S(p_i)$ of $p_i$ is appended to $S(p)$. The algorithm terminates when $S(p)$ is exhausted.
	\item Each time, when about to check the neighborhood of an object (line 13 in Algorithm~\ref{alg:expand}), the skipping rule is considered. If the rule applies, the algorithm continues to process the next object. The implementation of the skipping rule involves the testing of whether a circle is covered by the union of several circles. This can be accomplished using a recursive subdivision of the circle by non-overlapping squares~\cite{DBLP:journals/tods/WuYJ13}.
\end{itemize}

\subsection{Spatially Gridded Posting Lists}\label{sec:histogram}
In the previous section, we proposed a technique that reduces the number of objects to be exmained. However, for those objects that cannot be skipped, checking their neighborhoods involves time-consuming range queries. The result of checking a neighborhood is that the neighborhood is either dense or sparse. In this section, we design \textbf{spatially gridded posting lists} (SGPL) to estimate the selectivity of a range query on the IR-tree, such that sparse neighborhoods can be pruned without issuing expensive range queries.

An $n \times n$ grid is created on the data set. For each word $w$, a spatially gridded posting list is constructed covering all the objects that contain $w$. Let $D_{w_i}$ be the set of objects containing word $w_i$. The SGPL of $w_i$ is a sorted list of entries, where each entry takes the form  $(c_j,S_{w_i,c_j})$ where $c_j$ (sorting key) is the index value of a grid cell $C_{c_j}$ and $S_{w_i,c_j}$ is a set of objects that belong to $D_{w_i}$ and that are located in grid cell $C_{c_j}$, i.e., $\forall p \in S_{w_i,c_j}(p \in D_{w_i} \wedge p.\lambda \in C_{c_j})$.
Grid cells are indexed using a space filling curve, e.g., a Hilbert curve or a Z-order curve. The SGPLs of all distinct words in the data set are organized similarly to the inverted file. Empty cells are not stored. Given a word, its SGPL can be retrieved straightforwardly.

\begin{example}
	Figure~\ref{fig:sgpl}(a) illustrates a $4 \times 4$ grid on  the 7 objects in Example~\ref{ex:pex}. Grid cells are indexed using a 2-order Z-curve. Numbers in italics are the Z-order derived keys for the cells. 
	Figure~\ref{fig:sgpl}(c) shows the SGPLs for words `coffee' and `tea'. For example, the first entry in the SGPL for `coffee' tells that the document of $p_3$ contains `coffee' and $p_3$ is located in the cell with index value 3.
\end{example}

Given a set $q.\psi$ containing $m$ query keywords, the corresponding $m$ SGPLs are merged to estimate the selectivity of the range query. We define a merging operator $\bigoplus$ on several SGPLs that produces a count for each non-empty cell. The count for cell $C$ is the cardinality of the union of the sets of objects located in $C$ from different SGPLs, i.e.,
\begin{equation}
	\bigoplus_{w_i \in q.\psi}\mathit{SGPL}_{w_i}=\{(c_j,|\bigcup_{w_i \in q.\psi}S_{w_i,c_j} |\},
\end{equation}
where $q.\psi$ is the query keywords.

\begin{example}
	Consider the example SGPLs in Figure~\ref{fig:sgpl}(c). The third row is the result of merging the SGPLs of words `coffee' and `tea'. For example, the entry $(3,2)$ tells that the cell with index value 3 contains 2 objects after merging.
\end{example}

\begin{figure}[h]
	\begin{center}
		\begin{tabular}{c}
			\begin{tabular}{cc}
				\includegraphics[width=0.4\columnwidth]{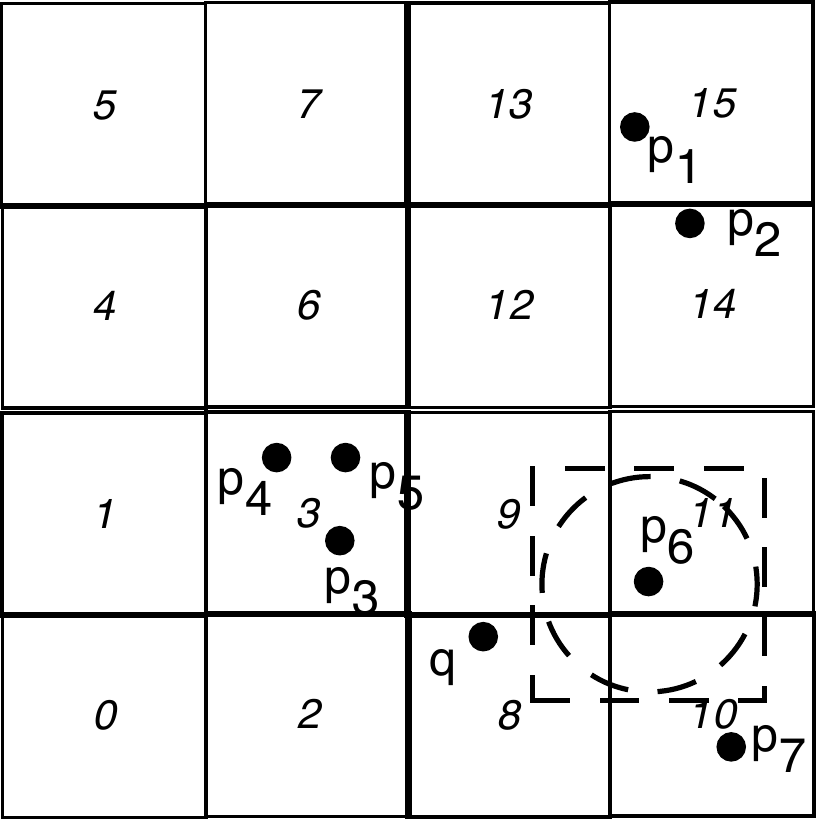}
				&  
				\includegraphics[width=0.4\columnwidth]{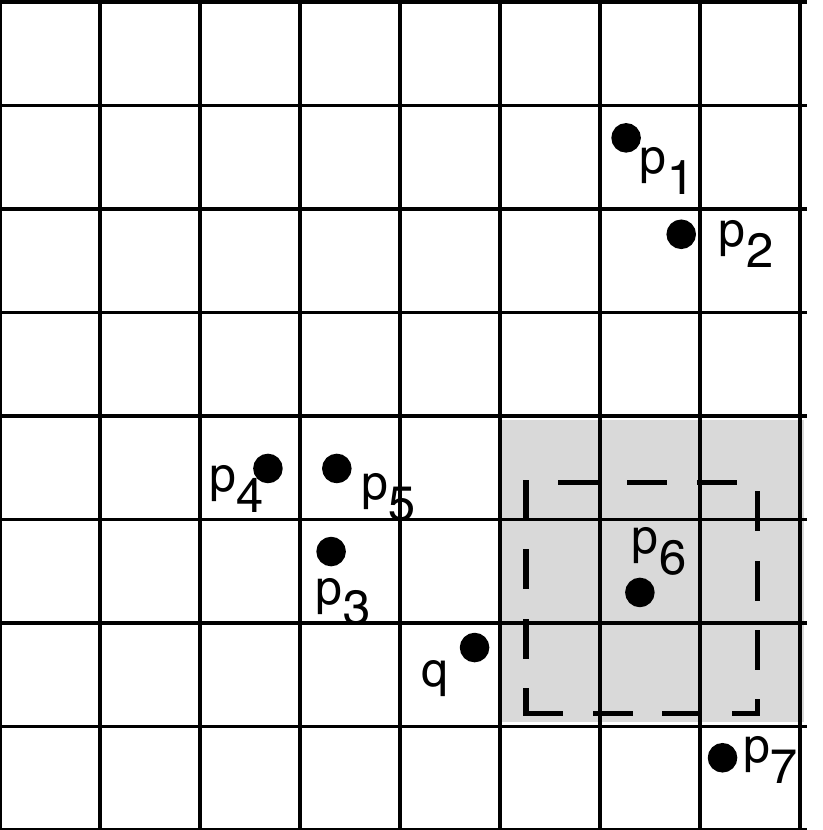} \\
				(a) $4 \times 4$ Grid &(b) $8 \times 8$ Grid
			\end{tabular}
			\\
			\\
			\small
			\begin{tabular}{|l|l|}
				\hline 
				coffee & $(3,\{p_3\}), (10,\{p_7\}), (11,\{p_6\}), (14,\{p_2\}),$ \\
				&$(15,\{p_1\})$\\
				\hline
				tea & $(3,\{p_5\}), (10,\{p_7\}), (14,\{p_2\}), (15,\{p_1\})$ \\
				\hline
				coffee $\bigoplus$ tea & $(3,2), (10,1), (11,1), (14,1), (15,1)$\\
				\hline
			\end{tabular}      
			\\
			(c) SGPL
		\end{tabular}
		\vspace{-2ex}
		\caption{Example Spatially Gridded Posting Lists}
		\label{fig:sgpl}
		\vspace{-4ex}
	\end{center}
\end{figure}

The merged result of the SGPLs of the query keywords is used to estimate the selectivity of the circular range query $q_c$ centered at an object $p$ with radius $\epsilon$ (e.g., the dashed circle in Figure~\ref{fig:sgpl}(a)). We approximate the circle $q_c$ as its circumscribed square $q_s$ (e.g., the dashed square in Figure~\ref{fig:sgpl}(a)). The sum of the counts of the grid cells that intersect square $q_s$ in the merged SGPLs of the query keywords is returned as the selectivity. Note that it is not necessary to merge the whole SGPL of each query keyword. For the sake of efficiency, only the cells that intersect $q_s$ need to be considered. We thus define a parameterized merging operator $\bigoplus(q_s)$ as follows.
\begin{equation}
	\bigoplus_{w_i \in q.\psi}(q_s)\mathit{SGPL}_{w_i}=\{(c_j,|\bigcup_{w_i \in q.\psi}S_{w_i,c_j} |~~|~~C_{c_j} \cap q_s \neq \emptyset \}
\end{equation}
Consider the example in Figure~\ref{fig:sgpl} where $q_s$ is the dashed square. We have $\mathit{coffee} \bigoplus (q_s) \mathit{tea}=\{(10,1), (11,1)\}$. Based on the coding scheme of the space-filling curve, we adopt an efficient existing algorithm~\cite{Lawder:2000:USC:646102.681186} to retrieve the cells that intersect the query range $q_s$.

The derived selectivity serves as an upper bound on the number of objects falling inside circle $q_c$.
If the selectivity (upper bound) is less than $q.\mathit{minpts}$, the $\epsilon$-neighborhood of object $p$ is guaranteed to be sparse. Otherwise, function $\mathbf{RangeQuery}$ is called to compute the exact number of objects in the $\epsilon$-neighborhood of object $p$. Hence, some sparse $\epsilon$-neighborhoods can be  found efficiently. However, it is not guaranteed that all sparse relevant $\epsilon$-neighborhoods can be found using the SGPLs of query keywords, since the selectivity is not always a tight upper bound due to the granularity of the grid and the use of the circumscribed square of the circular range.

\begin{example}
	Suppose $q.\mathit{minpts}=2$, meaning that a dense $\epsilon$-neighborhood should contain at least 2 objects. In Figure~\ref{fig:sgpl}(a), the dashed square $q_s$ approximates the $\epsilon$-neighborhood of  $p_6$. Since $q_s$ intersects grid cells 8, 9, 10, and 11 in the merged result of the SGPLs of words  `coffee' and `tea', the selectivity is 2, i.e., the number of objects covered by the four grid cells. Hence, the $\epsilon$-neighborhoods of $p_6$ is conservatively assessed as not sparse, and function $\mathbf{RangeQuery}$ has to be called to compute the exact number of objects in the $\epsilon$-neighborhood. However, if using the finer grid in Figure~\ref{fig:sgpl}(b), the selectivity is 1 ($< q.\mathit{minpts}$). Only the gray cells intersect the query range. Thus, the $\epsilon$-neighborhood of $p_6$ is guaranteed to be sparse, and there is no need to call function $\mathbf{RangeQuery}$.
\end{example}

We observe that using a finer grid may improve the quality of the selectivity estimation so that expensive $\mathbf{RangeQuery}$ operations are avoided. However, the cost of the selectivity estimation increases when using a finer grid. In the empirical study, we study how the performance is affected by the granularity of the grid.

\subsection{FastRange}\label{sec:fastrange}
In addition to supporting selectivity estimation, SGPLs are able to support the processing of range queries that are issued on the IR-tree. We propose an algorithm $\mathbf{FastRange}$ that handles range queries on SGPLs. Before presenting the algorithm, we first override the parameterised merging operator $\bigoplus(q_s)$ as $\overline{\bigoplus}(q_s)$.
\begin{equation}
	\overline{\bigoplus_{w_i \in q.\psi}}(q_s)\mathit{SGPL}_{w_i}=\{(c_j,\bigcup_{w_i \in q.\psi}S_{w_i,c_j} ~~|~~C_{c_j} \cap q_s \neq \emptyset \}
\end{equation}
The result of operator $\bigoplus(q_s)$ records the \textit{number} of objects inside each cell intersecting query range $q_s$, while the result of operator $\overline{\bigoplus}(q_s)$ contains the \textit{set of the identifiers} of the objects inside each cell intersecting query range $q_s$.
Algorithm~\ref{alg:fastrange} shows the pseudo code of $\mathbf{FastRange}$ that takes two arguments: $\mathit{list}$ is the result of operator $\overline{\bigoplus}(q_s)$, and $q_c$ is a circular region centered at an object $p$ and with radius $\epsilon$. If a cell $c$ from $\mathit{list}$ is completely inside the query range $q_c$, all the objects in $c$ are added to the result (lines 3 and 4). If a cell $c$ intersects $q_c$, only objects in $c$ that have distance to $p$ no greater than $\epsilon$ are added to the result (lines 6--8).
\begin{algorithm}[h]
	\caption{$\mathbf{FastRange}$(SGPL $\mathit{list}$, Range $q_c$)}\label{alg:fastrange}
	\begin{algorithmic}[1]
		\State $\mathit{result \leftarrow \emptyset}$;
		\For{each cell $c \in \mathit{list}$}
		\If{$c$ is completely inside the query range $q_c$}
		\State All the objects inside $c$ are added to $\mathit{result}$	
		\Else \Comment{$c$ intersects the query range}
		\For{each object $o$ inside $c$}
		\If{$\eucdist{o}{p} \leq \epsilon$}
		\State Add $o$ to $\mathit{result}$;
		\EndIf
		\EndFor
		\EndIf	
		\EndFor
	\end{algorithmic}
\end{algorithm}
\section{Empirical Studies}\label{sec:exp}
We conduct empirical studies to evaluate our proposals.
Section~\ref{sec:setup} presents the data set, queries, parameters, and
platform used in the experiments. The proposals for $k$-STC queries are evaluated in Section~\ref{sec:evaluation}.

\subsection{Experimental Setup}\label{sec:setup}
We use a real data set from TripAdvisor that contains 100,789 restaurants. Each restaurant has a text description of length 3 words on average. The total number of distinct words is 202. The data set is small. But as we are not aware of other public real data sets that  match our problem motivation, we generate larger data sets for scalability evaluation. Specifically, we generate data sets of size 200K, 400K, 600K, 800K, and 1M using TripAdvisor in the following way, such that the distribution of the real data set is roughly maintained. For a randomly picked object $p$ in TripAdvisor, we generate a new object whose location is obtained by slightly shifting the location of $p$ and whose text description is the same as that of $p$. 

We generate 4 query sets in the space of the dataset, in which the number of keywords
is 1, 2, 3, and 4, respectively. Each set comprises 100 queries. Queries are generated from
objects, and we guarantee that no query has an empty result.
Specifically, to generate a query, we randomly pick an object
in the dataset and take the location of the object as the query
location and randomly choose words from the document of
the object as the query keywords.

We evaluate the performance of the basic approach (Basic), the advanced approach with object skipping (Adv1), the advanced approach with object skipping and selectivity estimation (Adv2), and the advanced approach with object skipping, selectivity estimation, and FastRange (Adv3), under different parameter settings. Table~\ref{tab:parameter} shows the parameter values used in the experiments, where the bold values are default values.
All algorithms were implemented in Java, and an Intel(R) Core(TM) i7-3770 CPU @ 3.40GHz with 16GB main memory
was used for the experiments. All the data structures are memory resident. We report the average elapsed time
cost and the average number of range queries issued needed to compute the $k$-STC queries. 

\begin{table*}
	\caption{Parameter Values}\label{tab:parameter}
	\begin{center}
		\begin{tabular}{|c|c|c|}
			\hline
			Interpretation & Parameter & Values \\
			\hline
			Number of requested clusters & $k$ & 5, \textbf{10}, 15, 20 \\
			\hline
			Number of query keywords & $|q.\psi|$ & 1, \textbf{2}, 3, 4 \\
			\hline
			& $\epsilon$ & 0.0001, 0.0005, \textbf{0.001}, 0.005, 0.01 \\
			\cline{2-3}
			Density requirements & $\mathit{minpts}$ &  10, 20, \textbf{50}, 100, 200 \\
			\hline
			Order of Z-curve in SGPL & $h$ & 3, 4, 5, \textbf{6}, 7, 8, 9, 10 \\
			\hline
			Balancing the weights between  &&\\
			spatial distance	and text relevance & $\alpha$ & 0.1, 0.3, \textbf{0.5}, 0.7, 0.9 \\
			\hline
			Data set size & $|\mathcal{D}|$ & \textbf{100K}, 200K, 400K, 600K, 800K, 1M \\
			\hline
		\end{tabular}
	\end{center}
\end{table*}%

\subsection{Performance Evaluation}\label{sec:evaluation}
We evaluate our proposals against the basic algorithm under varying parameter settings. 

\stitle{Varying the Number of Keywords $|q.\psi|$}
Figure~\ref{fig:keyword} shows the performance of the four approaches when varying the number of query keywords. Their performance becomes worse as the number of query keywords increases, since more objects are involved so that the search space increases. The computation costs of the advanced algorithms increases more slowly than that of the basic algorithm. As expected, the number of range queries issued (Figure~\ref{fig:keywordrq}) is consistent with the elapsed time (Figure~\ref{fig:keywordcpu}).

\begin{figure}[h]
	\centering
	\centering
	\subfigure[Elapsed time]{\label{fig:keywordcpu}\includegraphics[width=0.45\columnwidth]{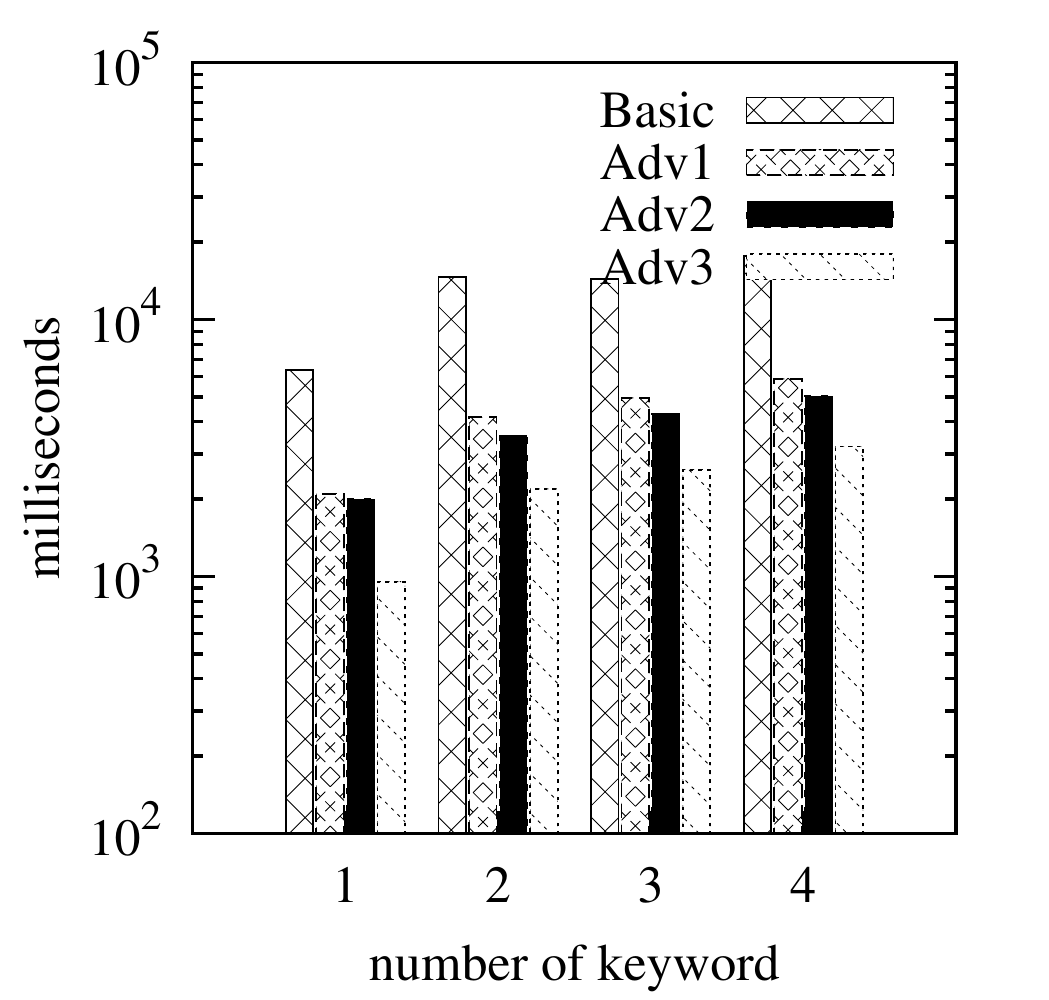}}
	\subfigure[$\#$ of range queries]{\label{fig:keywordrq}\includegraphics[width=0.45\columnwidth]{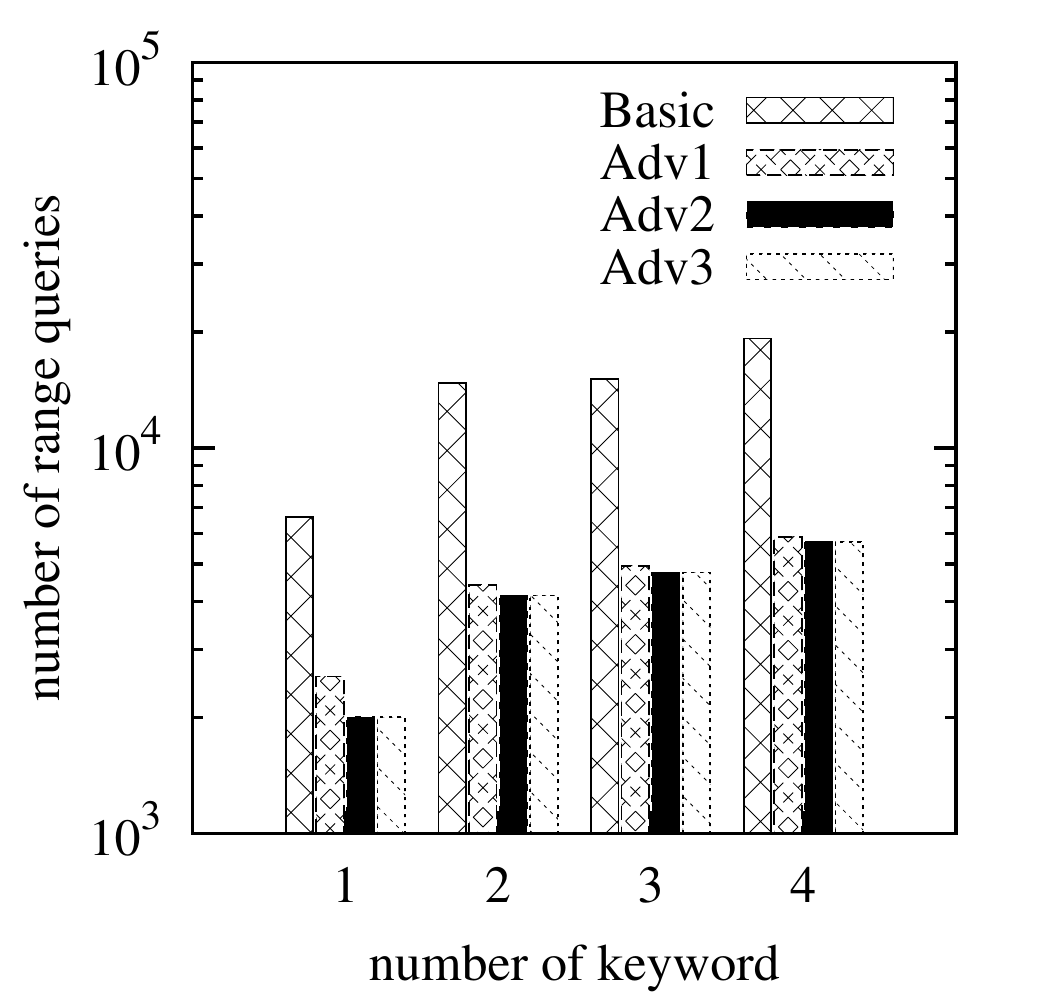}}
	\vspace{-2ex}
	\caption{Varying the Number of Keywords}\label{fig:keyword}
\end{figure}

\stitle{Varying Density Requirement}
The density requirement involves two parameters,  $\epsilon$ and $\mathit{minpts}$. Figure~\ref{fig:epsilon} shows the performance of the four approaches when varying $\epsilon$. Figure~\ref{fig:minpts} shows their performance when varying $\mathit{minpts}$. The $y$ axes are logarithmic. A large $\epsilon$ or a small  $\mathit{minpts}$ indicate a low density requirement. As $\epsilon$ increases, the performance becomes worse. The reason is that a low density requirement renders more objects' neighborhoods dense. Recall the basic idea of our proposals. If the $\epsilon$-neighborhood of a relevant object satisfies the density requirement, we need to examine the relevant objects inside the $\epsilon$-neighborhood to expand the cluster.  Hence, more dense neighborhoods incur additional computation cost. However, the performance is not sensitive to $\mathit{minpts}$ in range $[10,200]$ given that $\epsilon=0.001$. To altering the density requirement, we may choose to either vary $\epsilon$ or vary $\mathit{minpts}$. It is of no interest in finding the value range of $\mathit{minpts}$ that affects the performance for the interest of space.
\begin{figure}[h]
	\centering
	\subfigure[Elapsed time]{\label{}\includegraphics[width=0.45\columnwidth]{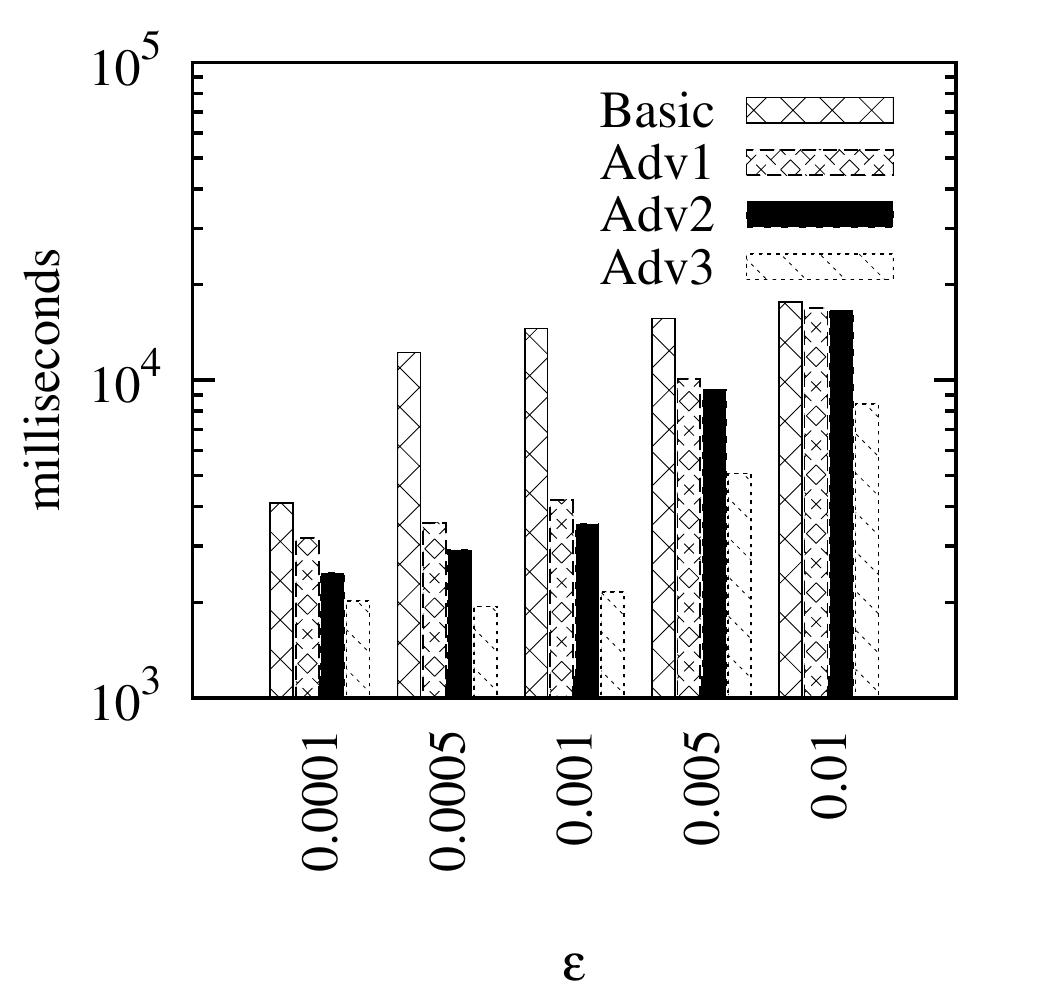}}
	\subfigure[$\#$ of range queries]{\label{}\includegraphics[width=0.45\columnwidth]{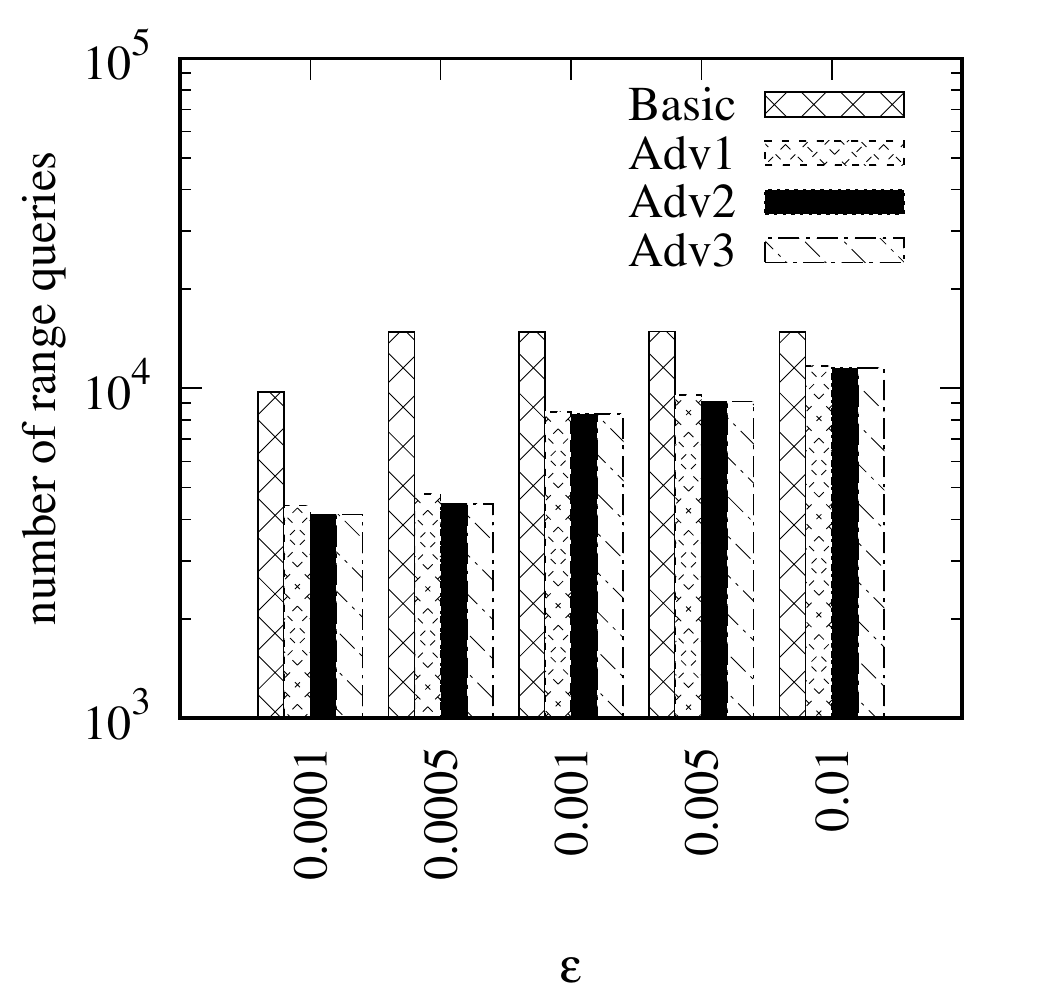}}
	\vspace{-2ex}
	\caption{Varying $\epsilon$}\label{fig:epsilon}
\end{figure}

\begin{figure}[h]
	\centering
	\subfigure[Elapsed time]{\label{}\includegraphics[width=0.45\columnwidth]{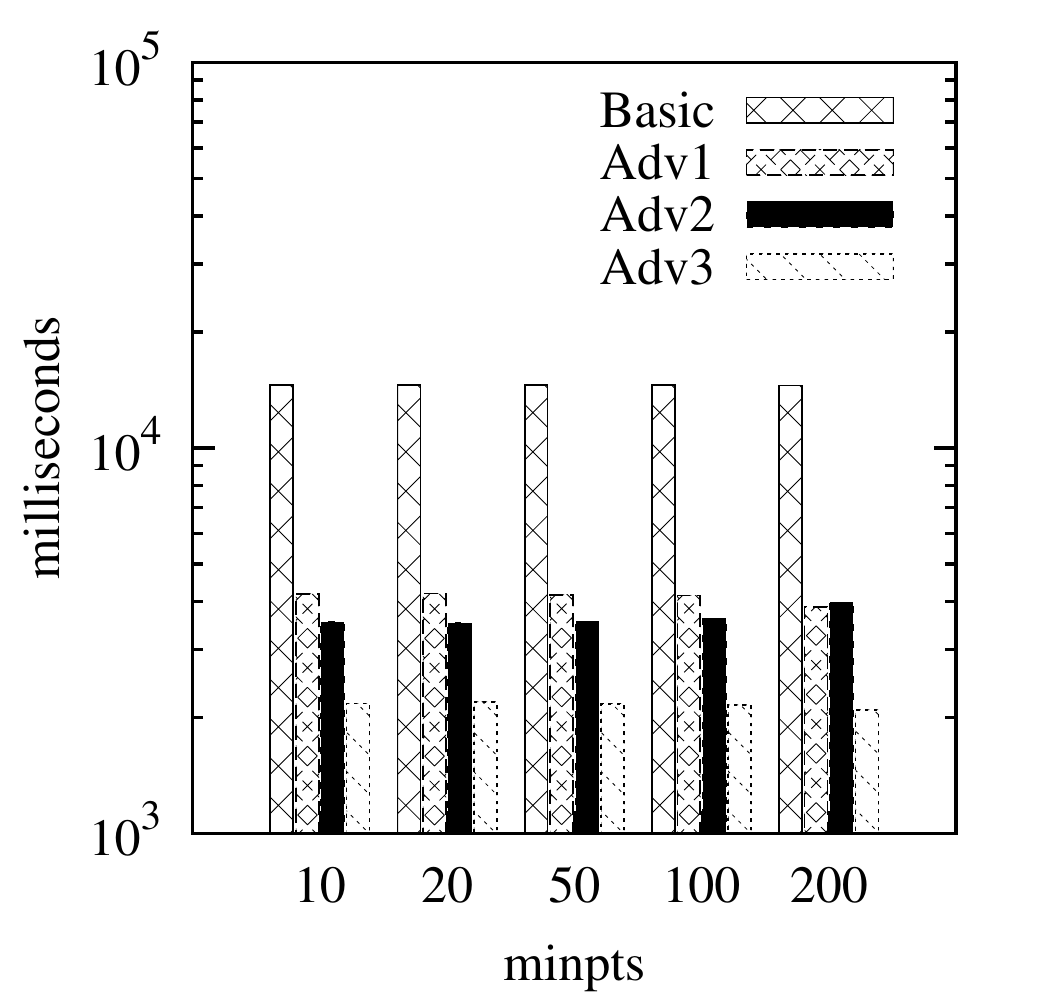}}
	\subfigure[$\#$ of range queries]{\label{}\includegraphics[width=0.45\columnwidth]{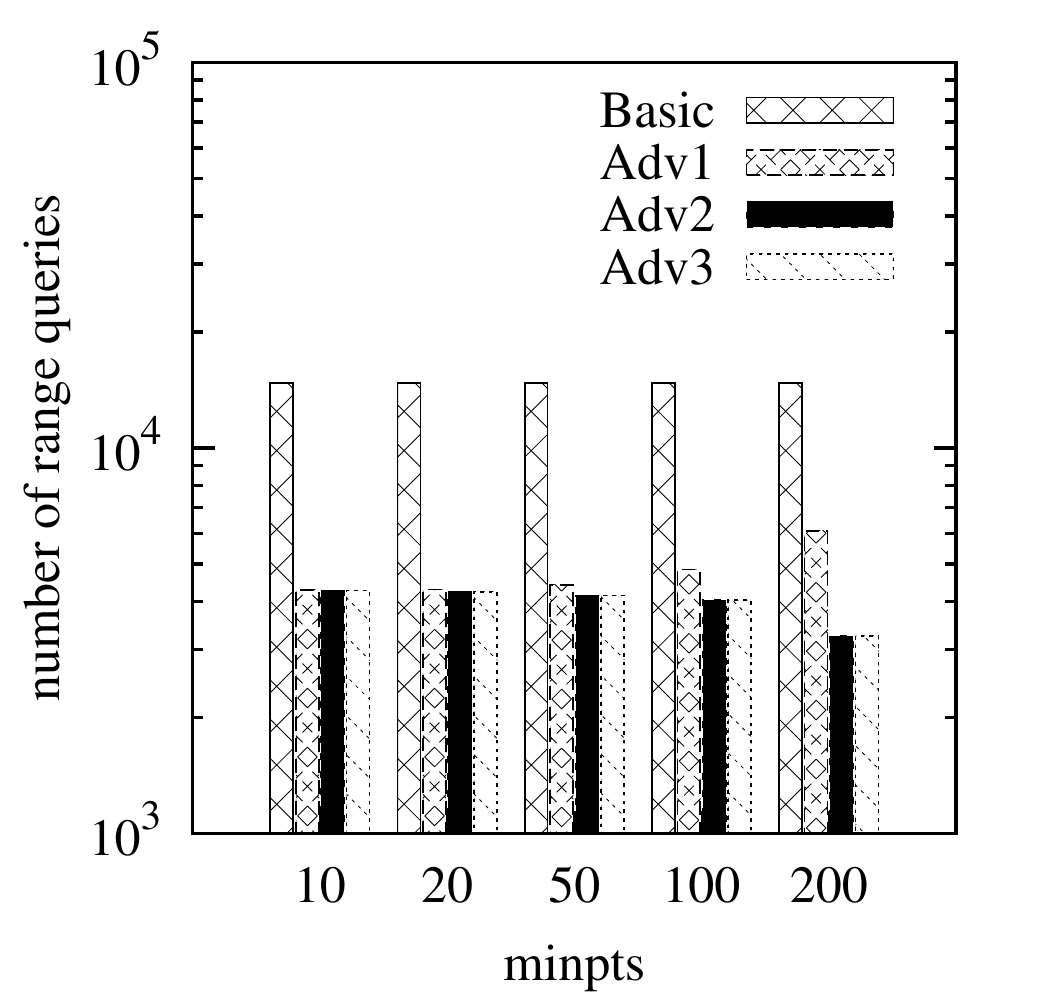}}
	\vspace{-2ex}
	\caption{Varying $\mathit{minpts}$}\label{fig:minpts}
\end{figure}

\stitle{Varying $\alpha$}
Figure~\ref{fig:alpha} shows the elapsed time and the number of range queries of the four approaches when varying $\alpha$ that balances the weight between the spatial distance and the text relevance in the scoring function. Their performance is insensitive to parameter $\alpha$. Hence, our approach can guarantee a good performance no matter what users' preferences are.

\begin{figure}[h]
	\centering
	\subfigure[Elapsed time]{\label{}\includegraphics[width=0.45\columnwidth]{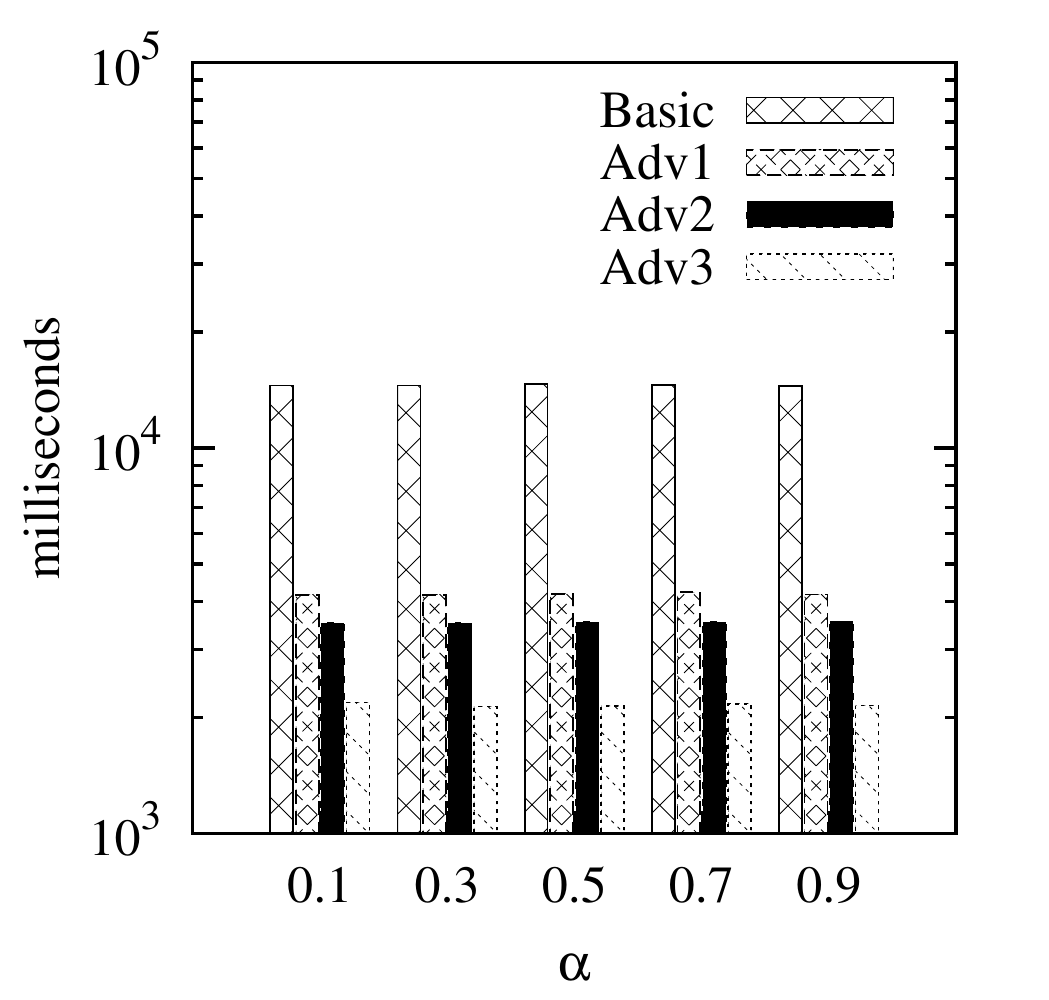}}
	\subfigure[$\#$ of range queries]{\label{}\includegraphics[width=0.45\columnwidth]{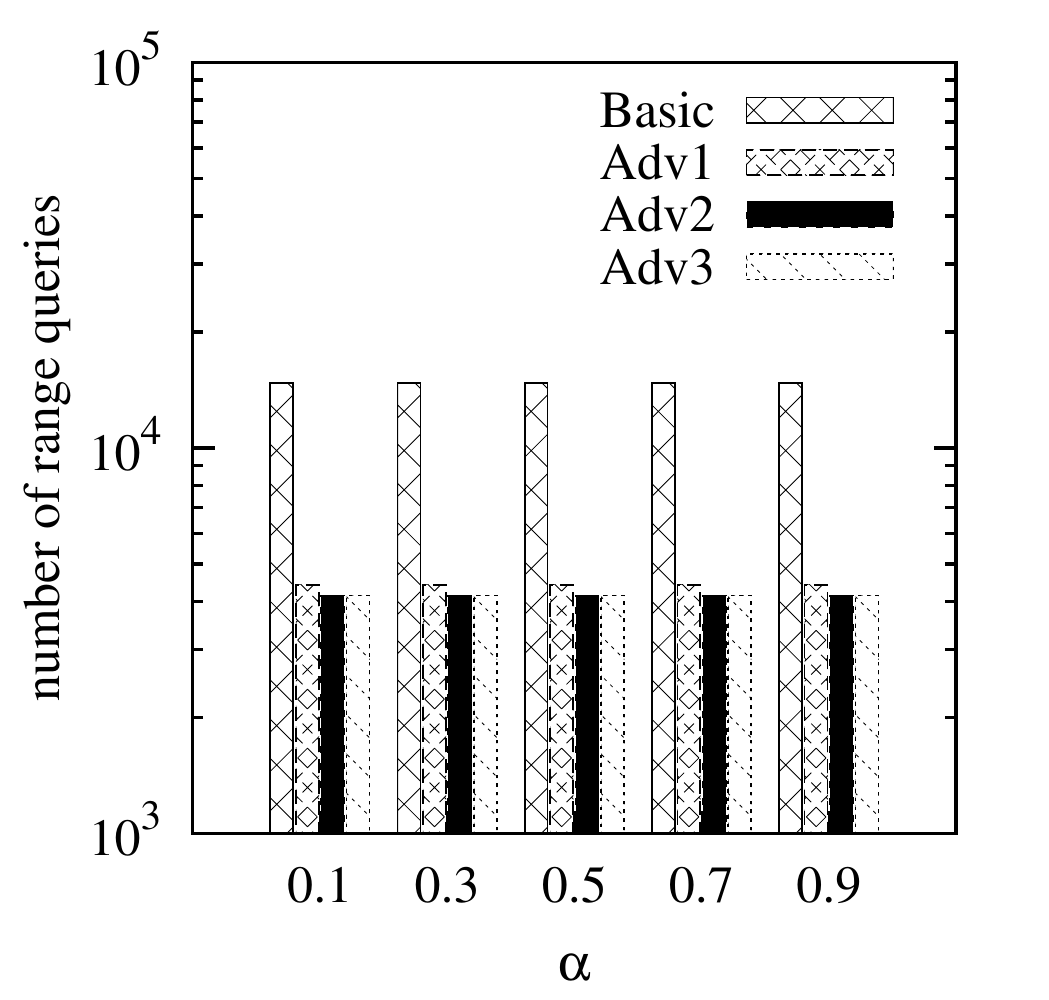}}
	\vspace{-2ex}
	\caption{Varying $\alpha$}\label{fig:alpha}
\end{figure}

\stitle{Varying the Number $k$ of Requested Clusters}
Figure~\ref{fig:topk} shows the elapsed time and the number of range queries of the four approaches when varying the number $k$ of requested clusters. Their performance is insensitive to parameter $k$. For a small $k$, the reason may be that the threshold in the basic algorithm is not tight enough to prune the clusters that are beyond the top-$k$ result early. For a large $k$, the reason may be that the total number of clusters found from the whole data set cannot exceed $k$. 

\begin{figure}[h]
	\centering
	\subfigure[Elapsed time]{\label{}\includegraphics[width=0.45\columnwidth]{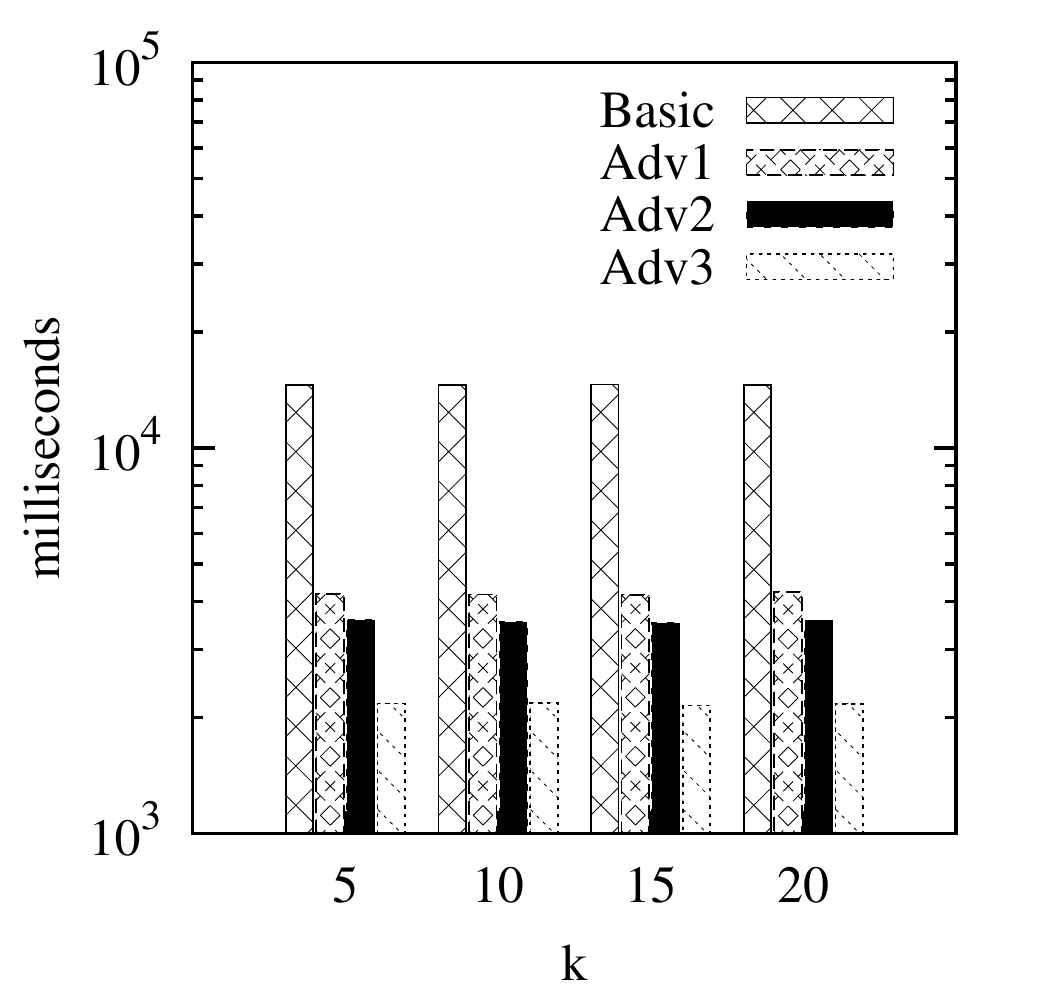}}
	\subfigure[$\#$ of range queries]{\label{}\includegraphics[width=0.45\columnwidth]{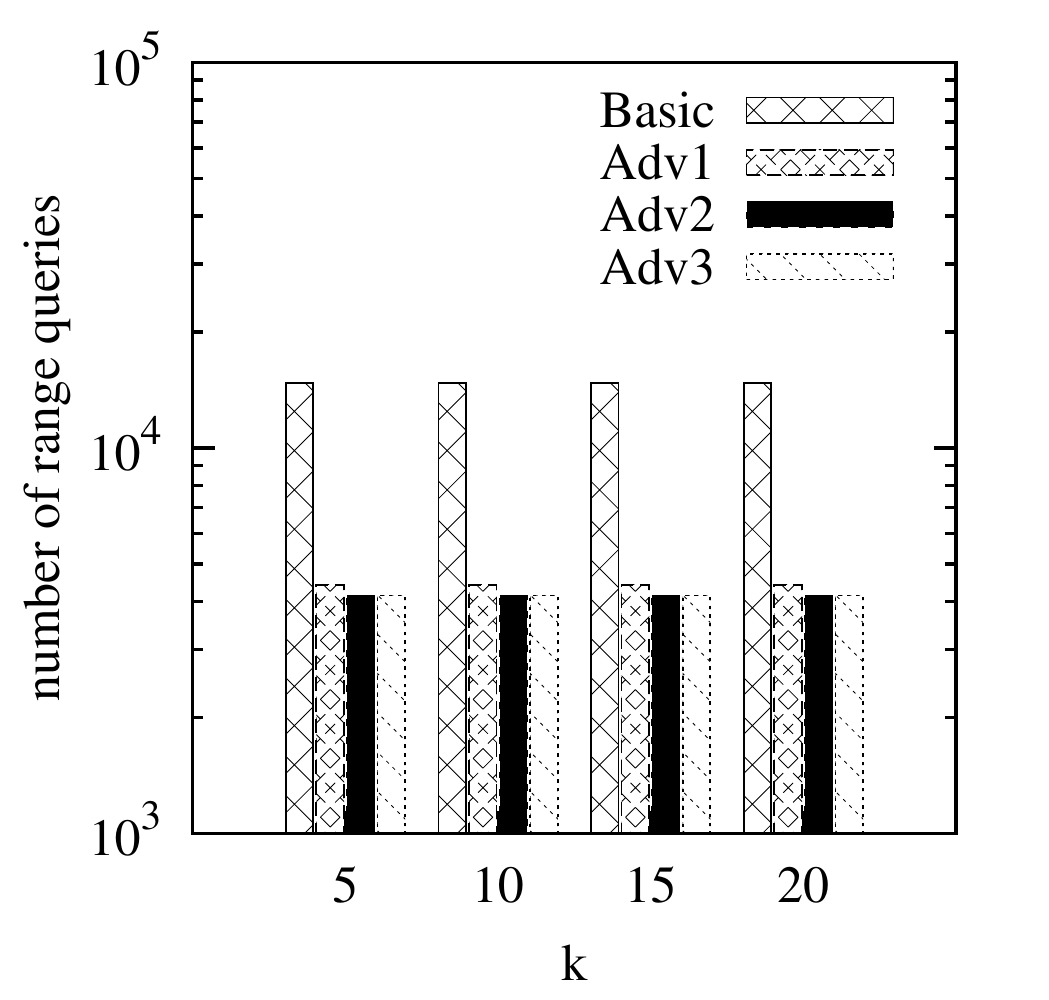}}
	\vspace{-2ex}
	\caption{Varying $k$}\label{fig:topk}
\end{figure}

\stitle{Scalability}
Figure~\ref{fig:datasize} shows how the performance of the basic algorithm and the advanced algorithm change as the size of the data set increases. We observe that both algorithms scale linearly. However, the computation cost of the advanced algorithm increases more slowly that does that of the basic algorithm.
\begin{figure}[h]
	\centering
	\subfigure[Elapsed time]{\label{}\includegraphics[width=0.45\columnwidth]{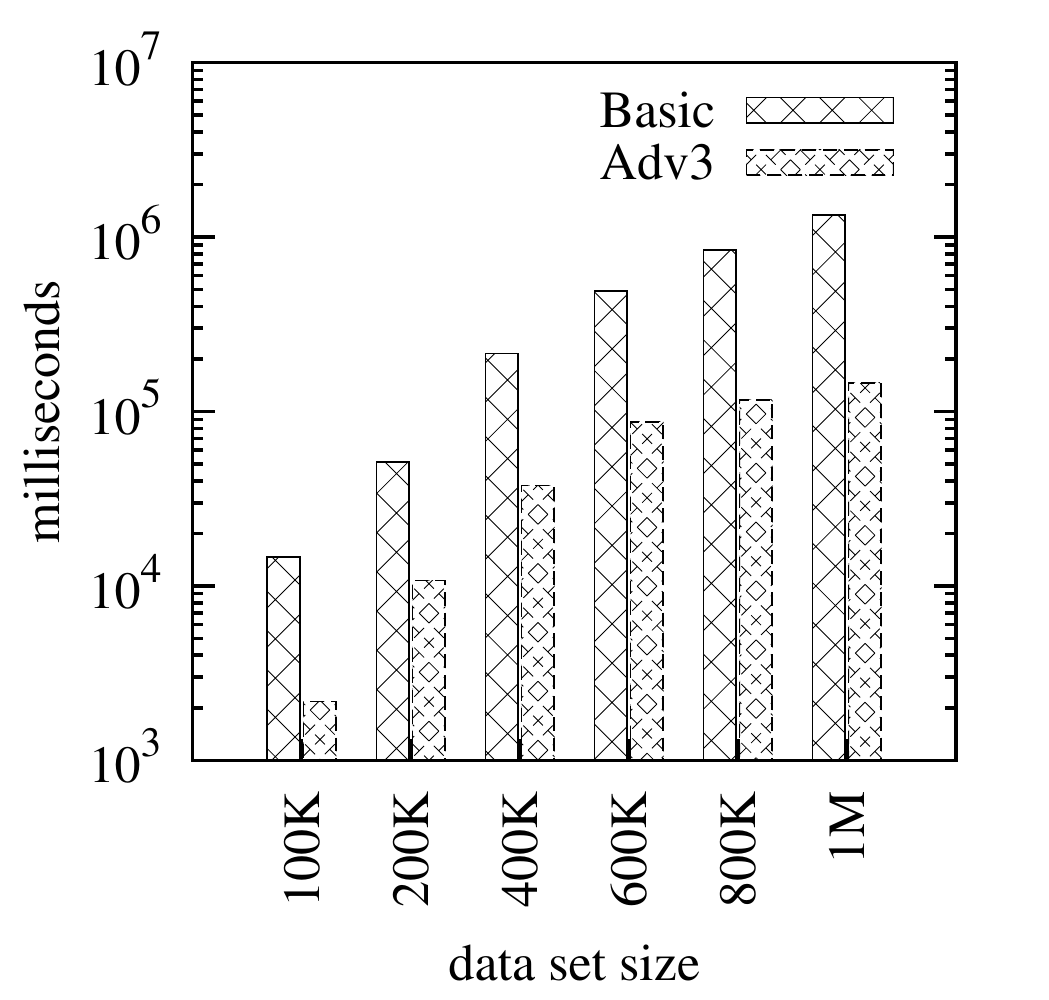}}
	\subfigure[$\#$ of range queries]{\label{}\includegraphics[width=0.45\columnwidth]{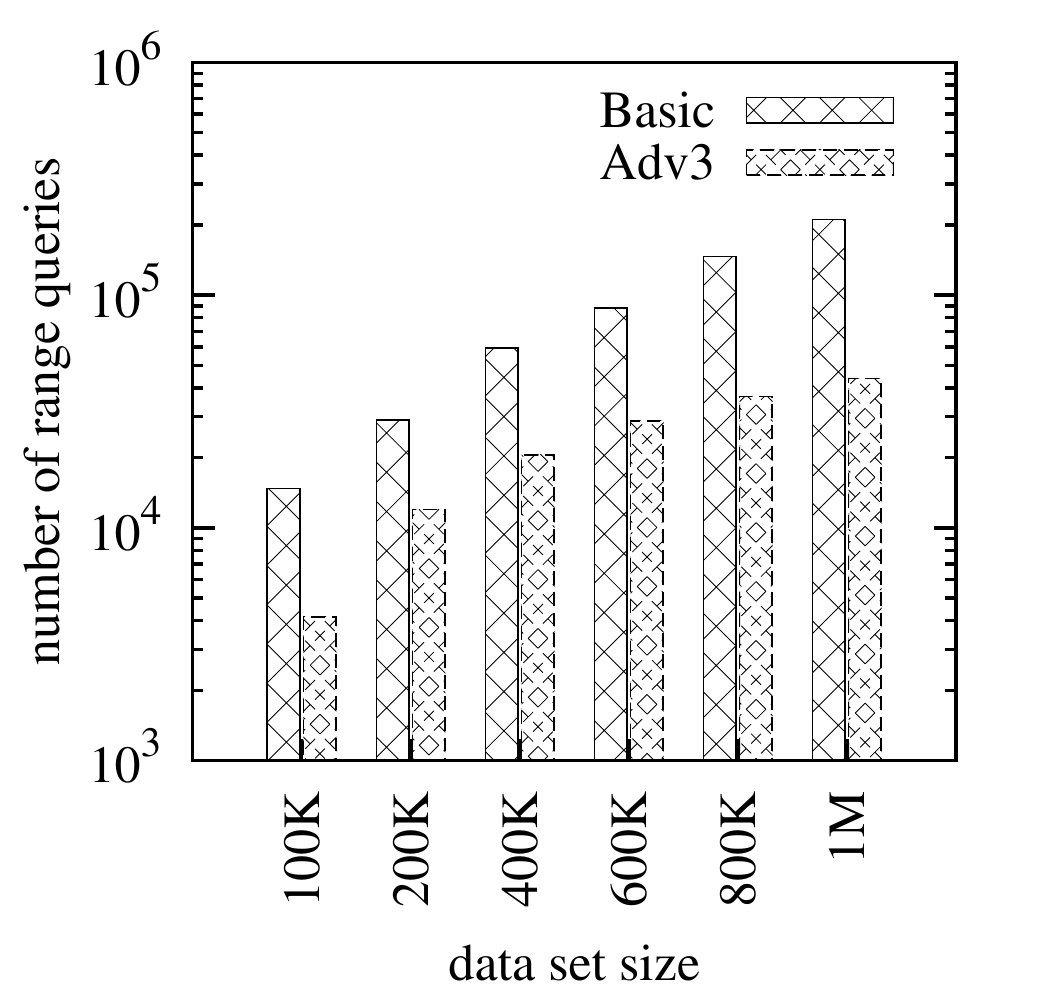}}
	\vspace{-2ex}
	\caption{Varying Data Size}\label{fig:datasize}
\end{figure}

\stitle{Varying the Order of Z-Curve $h$ in SGPL}
SGPL is constructed based on grid system over the data set, indexed by  a space filling curve. We adopt the Z-curve in our evaluation. Other types of space filling curves are also applicable. The order $h$ of the Z-curve defines the granularity of the grid system that affects the performance of SGPL. A large $h$ provides a finer granularity, so that the ability to estimate the selectivity of range queries and the ability to detect sparse $\epsilon$-neighborhoods are improved. The evaluation result in Figure~\ref{fig:kht} makes this point. As $h$ increases, the performance improves. We observe that when $h$ exceeds 6, the performance of SGPL becomes stable. This means that the selectivity estimation ability is close to optimal. It is not necessary to devote efforts to construct an SGPL using an even finer grid.

\begin{figure}[h]
	\centering
	\subfigure[Elapsed time]{\label{}\includegraphics[width=0.45\columnwidth]{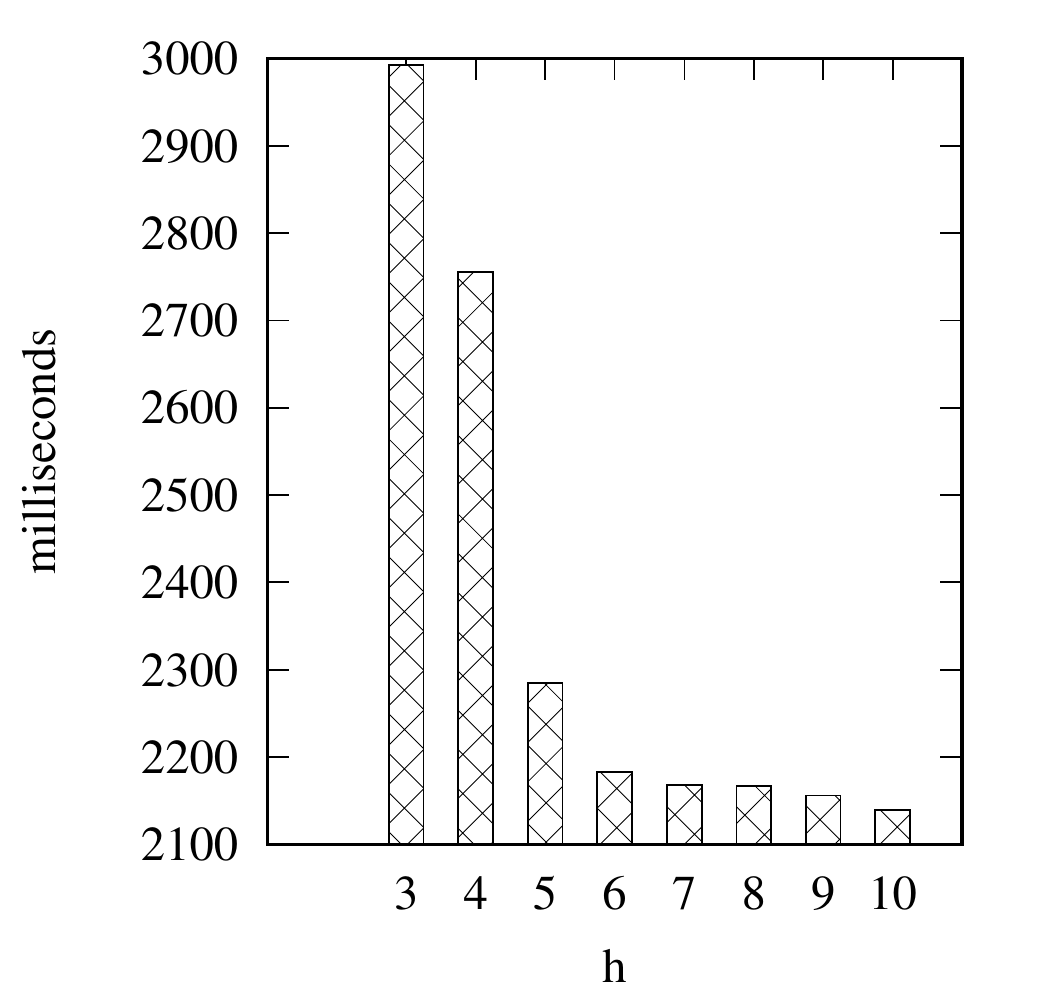}}
	\subfigure[$\#$ of range queries]{\label{}\includegraphics[width=0.45\columnwidth]{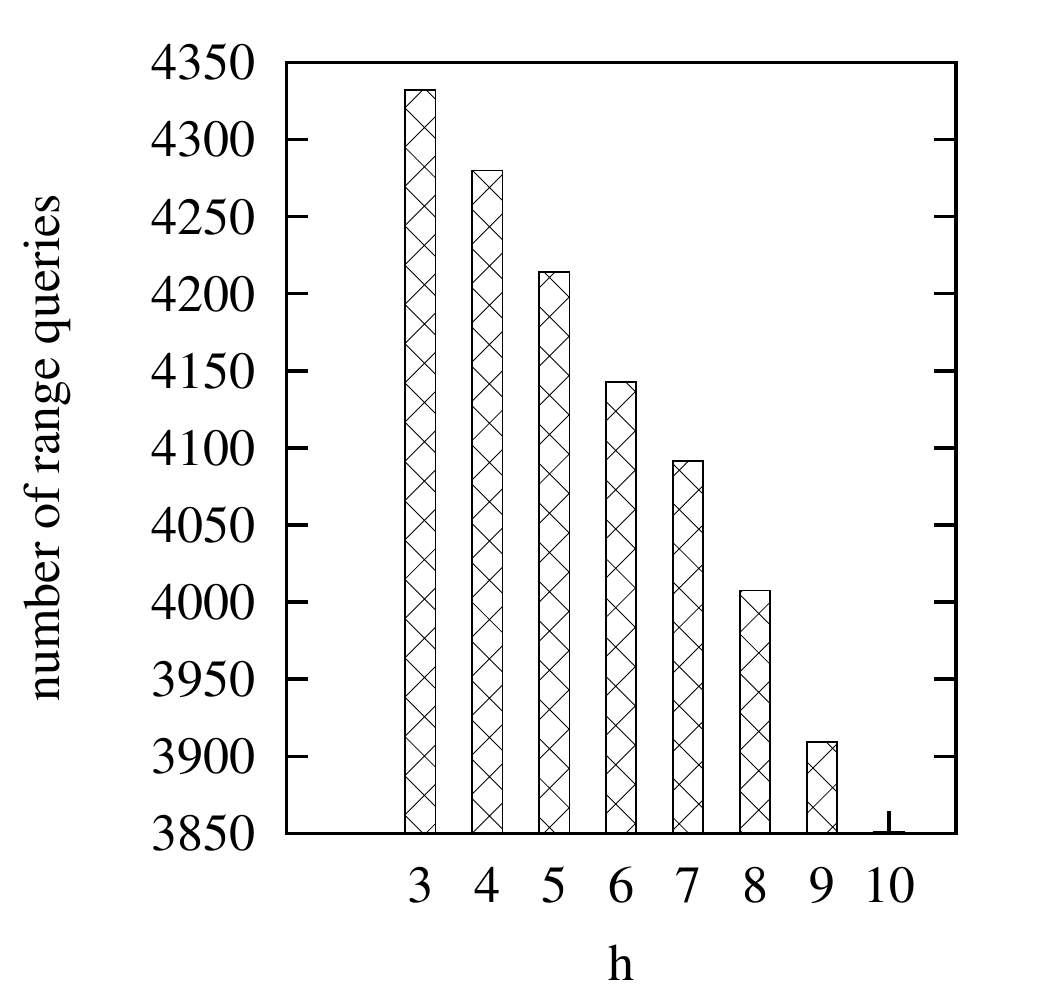}}
	\vspace{-2ex}
	\caption{Varying the Order $h$ in SGPL}\label{fig:kht}
\end{figure}

\stitle{Summary}
Overall, the best advanced approach outperforms the basic approach by an order of magnitude. The elapsed time and the number of range queries are proportional to each other, which indicates that the time consuming part of the $k$-STC queries is the range queries. The proposed advanced approach significantly reduce the cost incurred by the range queries.

\section{Related Work}\label{sec:related}
The problem studied in the paper is related to two topics: spatial textual search, surveyed in Section~\ref{sec:sts}, and density-based clustering, surveyed in Section~\ref{sec:clustering}.

\subsection{Spatial Textual Search}\label{sec:sts}

Spatial keyword search has attracted much attention in recent years. A spatial keyword query retrieves spatial web objects (e.g., web content about shops and restaurants) that are spatially and textually relevant to query arguments. Several efficient geo-textual indexes have been proposed to support spatial keyword queries. Most indexes combine the R-tree for spatial indexing and inverted files/signature files for text indexing, such as the IR$^2$-tree~\cite{DBLP:conf/icde/FelipeHR08}, the IR-tree~\cite{DBLP:journals/vldb/WuCJ12}, and the S2I~\cite{DBLP:conf/ssd/RochaGJN11}. These hybrid index structures are able to utilize both the spatial and textual information to prune the search space when processing a query. Our proposal is orthogonal to the above and similar  indexes. We adopt the  IR-tree~\cite{DBLP:journals/vldb/WuCJ12} in the basic algorithm. Other R-tree based indexes can also be used, and the resulting performance can be expected to be comparable to that of the basic algorithm.

Next, a range of studies investigate interesting variants of the prototypical spatial keyword query that satisfy different users' needs. Chen et al.~\cite{Chen:2013:EQI:2463676.2465328} study the problem of matching a stream of incoming Boolean range continuous queries against a stream of incoming geo-textual objects in real time. The keyword-aware optimal route query~\cite{Cao:2012:KOR:2350229.2350234} finds an optimal route that covers a set of user-specified keywords, meets a specified budget constraint, and achieves the best score according to a specific scoring function. The collective spatial keyword
query~\cite{DBLP:conf/sigmod/CaoCJO11,DBLP:journals/tods/CaoCGJO15,DBLP:conf/sigmod/LongWWF13} finds a group of objects that are close to a query point and that collectively cover a set of query keywords. Wu et al.~\cite{DBLP:journals/tods/WuYJ13,DBLP:conf/icde/WuYJC11} cover the problem of maintaining the result set of top-$k$ spatial keyword queries while the user (query location) moves continuously. The Location-aware top-$k$ Prestige-based Text retrieval (LkPT) query~\cite{DBLP:journals/pvldb/CaoCJ10} retrieves the top-$k$ spatial web objects ranked according to both prestige-based relevance and location proximity, where the prestige-based relevance captures both the textual relevance of an object to a query and the effects of nearby objects. The Reverse Spatial Textual $k$ Nearest Neighbor (RST$k$NN) query~\cite{DBLP:conf/sigmod/LuLC11} finds objects that take the query object as one of their $k$ most spatial-textual similar objects. 

Recent studies~\cite{DBLP:journals/pvldb/BoghSJ13,DBLP:journals/vldb/SkovsgaardJ15} target the common case where the user wishes to find nearby groups of points of interest that are relevant to the query keywords. Such groups are relevant to users who wish to conveniently explore several options before making a decision such as to purchase a particular product. The spatio-textual similarity join~\cite{DBLP:journals/pvldb/BourosGM12} retrieves the pairs of objects that are spatially close and textually similar.
%
%
%
Zhang et al.~\cite{DBLP:conf/icde/ZhangCMTK09,DBLP:conf/icde/ZhangOT10} study the query that takes a set of keywords and aims to find a set of geo-textual objects such that the union of their text descriptions cover all query keywords and such that the diameter of the objects is minimized. Considering the co-location relationship between objects, regions are retrieved so that the total weights of objects inside the result regions are maximized~\cite{DBLP:journals/pvldb/ChoiCT12,DBLP:journals/pvldb/TaoHCC13,DBLP:conf/cikm/LiuYS11}. However, these proposals impose specific shapes on the result regions, either fixed-sized rectangles or circles. A recent study~\cite{DBLP:journals/pvldb/CaoCJY14} computes a maximum-sum region where the road network distance between objects is less than a query constraint and the sum of the scores of the objects inside the region is maximized. This kind of query may retrieve a region containing many objects with low scores and ignore a promising region with few objects with high scores. Also, the query includes a query range that reduce the search space. The performance is unknown if the query range is set to cover  the whole data set.

It is well-known that clustering techniques are too expensive for interactive use. However, it is natural to try to leverage the work on clustering for the functionality that the paper targets. This paper is the first to explore the idea of performing clustering at query time. It thus explores a new direction in spatial keyword querying that represents an alternative to existing approaches. Our empirically study also demonstrates that it is indeed possible to adapt clustering to the paper's interactive setting. 

\subsection{Density-based Clustering Algorithms}\label{sec:clustering}

In density-based clustering, clusters are defined as regions with a higher density of objects than in the remainder of the  space covered by the data set. Objects in the sparse regions that separate clusters are usually considered to be noise and border objects. The most popular density based clustering method is DBSCAN~\cite{DBLP:conf/kdd/EsterKSX96}. It features a well-defined cluster model called ``density-reachability.'' It is based on connecting objects within a certain distance threshold of each other. A cluster consists of all density-connected objects  together with all objects that are within the threshold distance of these objects, and a cluster can have an arbitrary shape. OPTICS~\cite{Ankerst:1999:OOP:304182.304187} is a generalization of DBSCAN that removes the need to choose an appropriate value for the range parameter and that produces a hierarchical result related to that of linkage clustering. DBRS~\cite{DBLP:conf/pakdd/WangH03} improves DBSCAN by repeatedly picking an unclassified point at random and examining its neighborhood. VDBSCAN~\cite{4280175} aims for varied-density dataset analysis. Before adopting the traditional DBSCAN algorithm, some methods are used to select several values of parameter $\epsilon$ for different densities according to a $k$-dist plot. With different values of $\epsilon$, it is possible to find clusters with varied densities simultaneously. GDBSCAN~\cite{Sander:1998:DCS:593419.593465} clusters point objects as well as spatially extended objects according to both, their spatial and their nonspatial attributes. 

The techniques proposed in this paper are presented on the base of DBSCAN. However, it is possible to apply our approach to other density-based clustering models, since DBSCAN serves as the foundation for different variants of density-based clustering.

\section{Conclusions and Future Work}\label{sec:con}
This paper proposes a new type of query, namely the top-$k$ spatial textual clusters  ($k$-STC) query that returns the top-$k$ clusters that (i) are located the closest to the query location, that (ii) contain the objects whose text descriptions are the most relevant to the query keywords, and that (iii) have densities that exceed a threshold. Qualifying clusters are ranked according to a scoring function that takes into account the distance to the query location and the relevance to the query keywords. We propose a basic algorithm that is an on-line clustering method. To improve performance, we propose an advanced approach that includes three techniques: (i) a skipping rule that is used to reduce the number of objects to be examined, (ii) spatially gridded posting lists (SGPL) that are used to estimate the selectivity of the range queries so that sparse neighborhoods can be pruned at low cost, and (iii) a fast range query algorithm that leverages the SGPL.
Empirical studies on a real data set indicate that the paper's proposals are capable of excellent performance. 

This work opens to a number of interesting research directions that relate to the application of clustering to spatial-keyword querying. The basic algorithm contains an early stop condition that tries to obtain the top-$k$ clusters without finding all clusters from the whole data set. However, the tightness of the threshold in this early stop condition can be improved. It is of interest to derive a tighter threshold to save computation cost. The density requirements $\epsilon$ and $\mathit{minpts}$ in the query define the form of the clusters to be retrieved. However, users may not stick to a fixed density requirement, and the data set may have varying density, e.g., city center vs. countryside. A density requirement good for a city center may miss some results in the countryside. It is of interest to propose an algorithm that is able to automatically select suitable parameters for data sets with varying densities.


%

%
%

\end{document}